\documentclass[preprint,12pt,compress]{elsarticle}
\usepackage{lineno}
\usepackage{amsmath,amsfonts}
\usepackage{algpseudocode}
\usepackage{array}
\usepackage{textcomp}
\usepackage{stfloats}
\usepackage{url}
\usepackage{verbatim}
\usepackage{graphicx}
\usepackage{makecell}
\usepackage{epstopdf}
\usepackage{dcolumn}
\usepackage{bm}
\usepackage{natbib}
\setcitestyle{numbers,square}
\usepackage{braket}
\usepackage{color}
\usepackage{algorithm}
\usepackage{balance}
\usepackage{booktabs}
\usepackage{amssymb}
\usepackage{amsthm}
\usepackage{caption}
\usepackage{tikz}
\usepackage{adjustbox}
\usepackage{titlesec}
\usepackage{quantikz}
\usepackage{hyperref}

\theoremstyle{plain}
\newtheorem{theorem}{Theorem}

\newtheorem{lemma}{Lemma}
\newtheorem{proposition}{Proposition}
\newtheorem{corollary}{Corollary}

\theoremstyle{definition}

\newtheorem{remark}{Remark}

\makeatletter
\newenvironment{breakablealgorithm}
  {
   \begin{center}
     \refstepcounter{algorithm}
     \hrule height.8pt depth0pt \kern2pt
     \renewcommand{\caption}[2][\relax]{
       {\raggedright\textbf{\ALG@name~\thealgorithm} ##2\par}%
       \ifx\relax##1\relax 
         \addcontentsline{loa}{algorithm}{\protect\numberline{\thealgorithm}##2}%
       \else 
         \addcontentsline{loa}{algorithm}{\protect\numberline{\thealgorithm}##1}%
       \fi
       \kern2pt\hrule\kern2pt
     }
  }{
     \kern2pt\hrule\relax
   \end{center}
  }
\makeatother

\journal{}

\begin{document}

\begin{frontmatter}



\title{Distributed quantum approximate counting algorithm} 


\author[mymainaddress,mythirdaddress]{Huaijing Huang}
\author[mymainaddress,mythirdaddress]{Daowen Qiu\corref{mycorrespondingauthor}}
\ead{issqdw@mail.sysu.edu.cn}
\cortext[mycorrespondingauthor]{Corresponding author}

\affiliation[mymainaddress]{organization={School of Computer Science and Engineering},
            addressline={Sun Yat-sen University}, 
            city={Guangzhou},
            postcode={510006}, 
            country={China}}
\affiliation[mythirdaddress]{organization={The Guangdong Key Laboratory of Information Security Technology},
            addressline={Sun Yat-sen University}, 
            city={Guangzhou},
            postcode={510006}, 
            country={China}}
\begin{abstract}
 In this article, we design a distributed quantum algorithm for solving {the} counting problem using   {the} Grover operator and a classical post-processing procedure.   Compared to  existing counting algorithms, the proposed algorithm  has advantages in terms of the number of qubits, circuit depth, and the number of quantum gates. We apply the proposed algorithm to estimate inner products and Hamming distances. 
Simulations are conducted on the {Qiskit} platform, further demonstrating the effectiveness of our algorithm and its suitability for the NISQ era.  

\end{abstract}


\end{frontmatter}



\section{Introduction}
The counting problem is a fundamental issue in computer science. Given $N\in \mathbb{Z}^+$ elements, among which there are $m>0$ marked elements, the counting problem involves estimating the number of marked elements. The quantum counting algorithm was originally proposed by Brassard et al. in ~\cite {brassard1998quantum,brassard2000quantum}, and its query complexity is $O(\sqrt{N})$. An approximate quantum counting algorithm was also proposed in~\cite {brassard2000quantum}, with a query complexity of $O(\frac{1}{\epsilon}\sqrt{\frac{N}{m}})$, where $0<\epsilon<1$ represents the error. Similar to Grover's algorithm~\cite {Grover1996}, {the} quantum counting algorithm achieves a quadratic speedup in query complexity compared to classical algorithms.

Quantum counting algorithms {estimate the number of marked elements by
estimating the amplitude associated with the marked subspace.}  Therefore,  {the} quantum amplitude estimation algorithms are closely related to quantum counting algorithms. The quantum amplitude estimation algorithm in ~\cite {brassard2000quantum} actually combines quantum phase estimation~\cite {kitaev1995quantum} and Grover's algorithm. Recently, many quantum amplitude estimation algorithms have emerged that do not rely on quantum phase estimation {involving the quantum} Fourier transform; they only require the use of {the} Grover {operator} and do not need to rely on {the} quantum Fourier transform,  making them more suitable for NISQ ({noisy} intermediate-scale quantum)~\cite {Preskill2018} devices. These quantum amplitude estimation algorithms include {the methods based on} maximum likelihood estimation~\cite {suzuki2020amplitude}, interval estimation~\cite {grinko2021iterative,Fukuzawa2023,zhao2022adaptive,manzano2023real}, and variational {approaches}~\cite {plekhanov2022variational}, among others. 

Tan et al.~\cite {tan2017} generalized the algorithm from \cite {brassard2000quantum} and proposed a generalized quantum counting algorithm with an initial state that is a non-uniform superposition. The quantum counting algorithm has recently been further advanced, evolving from {relying} on the phase estimation algorithm {involving the quantum} Fourier transform to  requiring only {the} Grover {operator}.  Wie ~\cite {wie2019simpler} proposed a simple quantum counting algorithm based on  {a Grover operator controlled by a single qubit}. 
A quantum approximate counting algorithm was proposed by Aaronson et al.~\cite {aaronson2020quantum} using  {the} Grover {operator}, accompanied by  rigorous analyses.
 Venkateswaran et al.~\cite {venkateswaran2020quantum} introduced a non-adaptive Grover iterations quantum approximate counting algorithm. Le Gall et al.~\cite {gall2022quantum} extended  the quantum approximate counting technique to estimate the number of marked states in a Markov chain. Notably, an  interesting error-correction method suitable for designing distributed quantum algorithms was proposed by Qiu et al.~\cite {qiu2025universal}.
  
{According to the definition of distributed quantum computing in \cite{Li2024}, distributed quantum computing is, in a way,  to decompose a computational problem into
smaller subproblems that are executed by multiple quantum computing
nodes, which collaboratively solve the original problem through quantum
or classical communication}. Various distributed quantum
algorithms have been proposed, including distributed Grover's
algorithm~\cite {Qiu2022}, distributed Simon's
algorithm~\cite {TAN_quantum_2022}, distributed Shor's
algorithm~\cite {Xiao2023}, and 
other related approaches~\cite {Li2025DistributedDJ,  Li2024, Huang2026inpress, HQ2025}.{These algorithms are based on different communication and implementation frameworks.
 Some methods require shared entanglement or quantum
communication among quantum computing nodes~\cite {TAN_quantum_2022,Xiao2023,Li2025DistributedDJ,Li2024}, whereas others
allow the nodes to perform their quantum computations independently and
exchange only classical information~\cite {qiu2025universal, Qiu2022, Huang2026inpress, HQ2025, Avron2021}. To the best of our knowledge,
no distributed quantum counting algorithm that operates without quantum
communication between quantum computing nodes has yet been proposed.} 

{The present work is motivated by the limitations on both the number of
qubits and the depth of quantum circuits that can be reliably executed on
an individual quantum processor. For a given value of \(k\), the counting
task is decomposed into \(2^k\) subtasks that can be executed in parallel
on \(2^k\) quantum computing nodes, with each node requiring only
\(n-k+2\) qubits. Thus, several quantum processors with smaller quantum
registers can be used instead of a single processor with a larger quantum
register. When only a single small-scale quantum processor is available,
the subtasks can instead be executed sequentially using the same
\(n-k+2\) qubits.}

{Compared with distributed quantum algorithms that require quantum
communication~\cite {TAN_quantum_2022,Xiao2023,Li2025DistributedDJ,Li2024}, the proposed framework has zero quantum communication cost, requires no shared entanglement among the quantum computing nodes, and exchanges only classical task information and counting estimates with the central classical computer.  The proposed
framework  is designed specifically for quantum counting, and the final count is obtained by summing the estimates of the disjoint
counting subtasks. A common implementation assumption of these methods is
that each node must be able to realize the problem-dependent quantum
operation associated with its assigned subproblem. In our
framework, the \(j\)-th node implements the  oracle
\(U_{\chi_{S_j}}\). Since no 
shared entanglement or coherent operation across different quantum
processors is required, the proposed framework has less demanding
inter-node implementation requirements than methods that rely on quantum
communication or shared entanglement.}

{In terms of scalability, the parameter \(k\) controls the trade-off
between the number of quantum computing nodes and the resources required
by each node. With \(2^k\) nodes, each node requires \(n-k+2\) qubits,
and \(k\) can be selected according to the qubit capacity and the number
of available quantum computing devices. Increasing \(k\) reduces the
per-node qubit requirement but increases the number of counting tasks.
Moreover, the comparison with the centralized quantum counting algorithms
considered in this work shows that the proposed framework can reduce the
number of qubits, circuit depth, and number of quantum gates required for
the circuit executed by each quantum computing node.  Consequently, the framework is
suitable for settings in which multiple resource-constrained quantum
processors are available but quantum communication between them is
unavailable.}

{Quantum counting has been applied to proportion
estimation~\cite {Diao2012QuantumCounting}, model
counting~\cite {Riguzzi2024QuantumAlgorithms}, and privacy-preserving
set-cardinality computation and condition
querying~\cite {9726782}. Closely related quantum amplitude estimation
techniques have also been used for Monte Carlo-type
estimation~\cite {Montanaro2015QuantumSpeedup}. In addition,
methods based on quantum counting have been developed for distributed
correlation and Hamming-distance estimation, as well as distributed
machine-learning tasks~\cite {Hao2023}. Accordingly, the proposed algorithm may be adapted to distributed settings in which the search space or dataset is partitioned into disjoint subsets, each processed by a separate quantum node, and the final estimate is obtained by aggregating the results produced by these nodes.  Potential
applications include distributed database querying and set-cardinality
estimation, model counting for Boolean formulas, Monte Carlo-type estimation
problems that can be reduced to counting marked outcomes, and distributed
learning tasks involving correlation or Hamming-distance estimation.}

The paper is structured in the following manner. In Section \ref{2}, {the} modified iterative quantum amplitude estimation (MIQAE) algorithm  will be reviewed briefly. In Section \ref{3},  we design a  distributed quantum counting algorithm to approximate the number of marked elements in a given set.   We also prove the correctness of the algorithm,  analyze the query complexity,  and compare it with existing work. The two applications of the algorithm are presented in Section \ref{4}. Using the  {Qiskit} platform, {we verify} the effectiveness of the algorithm in  Section \ref{experiment1}.  Conclusions and  prospects are presented in Section \ref{6}.

\section{Preliminaries}\label{2}
Iterative quantum amplitude estimation was initially proposed by Grinko et al. in 2021~\cite {grinko2021iterative}, and {was} later improved by Fukuzawa et al.~\cite {Fukuzawa2023}. In this section,  we  specifically introduce {the} modified iterative quantum amplitude estimation (MIQAE) algorithm ~\cite {Fukuzawa2023}.
 \subsection{The modified iterative quantum amplitude estimation  algorithm}
 The quantum amplitude estimation problem is as follows: given a unitary operator $\mathcal{A}$ acting on $n+1$ qubits, such that \begin{equation}\label{def-A}
    \mathcal{A}\ket{0}_n \ket{0} = \cos\theta_a\ket{\psi_0}_n \ket{0} + \sin\theta_a\ket{\psi_1}_n \ket{1},
\end{equation}
where $\theta_a \in [0,\pi/2]$,  $\ket{\psi_0}$ and  $\ket{\psi_1}$ are two normalized states. The subscript \( n \) denotes the number of qubits. The goal is to estimate $\sin^2\theta_a$ by making as few calls to $\mathcal{A}$ as possible. 
Let $a=\sin^2\theta_a$. We define {the} Grover operator as follows: \begin{equation}Q = -\mathcal{A}\mathcal{U}_0\mathcal{A}^\dagger \mathcal{U}_{\psi_1},\end{equation} where $\mathcal{U}_0 = \mathbb{I}_{n+1} - 2 \ket{0}_{n+1}\bra{0}_{n+1}$ and $\mathcal{U}_{\psi_1} = \mathbb{I}_{n+1} - 2 \ket{\psi_1}_{n}\bra{\psi_1}_{n}\otimes \ket{1}\bra{1}$. The iterative quantum amplitude estimation algorithm does not require controlled applications of  $Q$. Instead, it directly applies the $Q$ operator \(k\) times to $\mathcal{A}\ket{0}_n \ket{0}$, resulting in $Q^k\mathcal{A}\ket{0}_n \ket{0}$,  and then repeatedly measures the last qubit. The amplitude estimate is determined by the probability of measuring $\ket{1}$, that is \begin{equation}
    \mathbb{P}[\ket{1}] = \sin^2((2k+1)\theta_a). 
\end{equation}

The MIQAE algorithm primarily involves specifying a confidence interval $ [\theta_l, \theta_u]$ {for $\theta_a$} and utilizing Algorithm  \ref{alg:find-next-k} to find distinct values of $K_i$, where $K_i = 2 k_i + 1$,  such that the confidence interval can progressively shrink to the desired precision {with the prescribed confidence level}. Ultimately, this process estimates the amplitude $\sin^2\theta_a$. The core idea of Algorithm \ref{alg:find-next-k} is to find the largest odd integer $K$ in the range from $K_i$ to $\frac{\pi}{2(\theta_u-\theta_l)}$ such that both $K\theta_l$ and $K\theta_u$ lie in the same quadrant, using the current values of $k_i$, $\theta_l$, and $\theta_u$. The  $K_i$ identified by Algorithm \ref{alg:find-next-k} must satisfy the following condition \begin{equation}\label{11}
    \left\lfloor \frac{K_i \theta_l}{\pi/2} \right\rfloor = \left\lceil \frac{K_i \theta_u}{\pi/2} \right\rceil - 1.
\end{equation}
The quadrant count is defined as $R_i = \lfloor K_i \theta_l / (\pi/2) \rfloor$, which starts from 0. Algorithm \ref{alg:find-next-k} must find a $K_i$ such that the enlarged interval $[K_i\theta_l, K_i\theta_u]$ falls within one quadrant; otherwise, it returns the previous {value} $K_{i-1}$.
According to Eq. (\ref{11}), $K_i \theta_l$ and $K_i \theta_u$ can be expressed as follows:\begin{align*}
    K_i \theta_l & = R_i \cdot (\pi/2) + \gamma^{min},& \mbox{for}~\gamma^{min} \in [0, \pi/2), \\
    K_i \theta_u & = R_i \cdot (\pi/2) + \gamma^{max},& \mbox{for}~\gamma^{max} \in (0, \pi/2].
\end{align*}
With repeated measurements, we can obtain the confidence interval $[a^{min}, a^{max}]$ for $\sin^2 (K_i \theta_a)$. Based on this interval and Table \ref{table:theta-i-ci}, we can determine the confidence interval for $\gamma_a$, which is $[\gamma^{min}, \gamma^{max}]$. We utilize the Chebyshev-Hoeffding inequality to estimate the probability that $\sin^2 (K_i \theta_a)$ falls outside the confidence interval $[a^{min}, a^{max}]$, i.e., $\sin^2 (K_i \theta_a) \not\in [a^{min}, a^{max}]$.
\begin{table}[!t]
\centering
\begin{tabular}{|c|c|c|}
\hline
Quadrant ($R$)& $\gamma_i^{min}$                   & $\gamma_i^{max}$                   \\ \hline
Even        & $\arcsin \sqrt{a_i^{min}}$         & $\arcsin \sqrt{a_i^{max}}$         \\ \hline
Odd       & $-\arcsin \sqrt{a_i^{max}} + \frac{\pi}{2}$  & $-\arcsin \sqrt{a_i^{min}} + \frac{\pi}{2}$  \\ \hline
\end{tabular}
\caption{\em A table describing the conversion from $a_i^{min}$, $a_i^{max}$ to $\gamma_i^{min}$, $\gamma_i^{max}$ respectively. $\sin^{-1}$ denotes the arcsin function. The top right and bottom left quadrants correspond to an even quadrant count. The top left and bottom right quadrants correspond to an odd quadrant count.}
\label{table:theta-i-ci}
\end{table}

\begin{algorithm}
\caption{Modified IQAE}
\label{alg:modified-iqae}
\begin{algorithmic}[1]
\Require $\epsilon > 0$, $\alpha > 0$, $N_{\text{shots}} \in \mathbb{N}$, Unitary $\mathcal{A}$
\State $i = 0$, $k_i = 0$ 
\State $[\theta_l, \theta_u] = [0, \pi/2]$ 
\State $K_{max}= \frac{\pi}{4\epsilon}$ 
\While{$\theta_u - \theta_l > 2\epsilon$}
\State $N = 0$, $i = i + 1$, $k_i = k_{i-1}$
\State $K_i = 2k_i + 1$
\State $\alpha_i = \frac{2\alpha}{3} \frac{K_i}{K_{max}}$
\State $N_i^{max} = \frac{2}{\sin^2(\pi/21)\sin^2(8\pi/21)}\ln\left(\frac{2}{\alpha_i}\right)$
\State \begin{varwidth}[t]{.95\linewidth} 
$R_i = \lfloor \frac{K_i\theta_l}{\pi/2} \rfloor$
\end{varwidth}
\While{$k_i = k_{i-1}$}
\State Prepare circuit $Q^{k_i}\mathcal{A}\ket{0}_n\ket{0}$
\State Measure $\min\{N_{\text{shots}}, N_{i}^{max}-N\}$ times 
\State \begin{varwidth}[t]{.95\linewidth} Combine results in this round to approximate $\hat{a}_i \approx \mathbb{P}\left[\ket{1}\right]$. \end{varwidth}
\State $N = N + \min\{N_{\text{shots}}, N_{i}^{max}-N\}$
\State $\epsilon_{a_i}= \sqrt{\frac{1}{2N} \ln\left(\frac{2}{\alpha_i}\right)}$ 
\State $a_i^{\max} = \min(1, \hat{a}_i + \epsilon_{a_i})$
\State $a_i^{\min} = \max(0, \hat{a}_i - \epsilon_{a_i})$
\State \begin{varwidth}[t]{.88\linewidth} Calculate the confidence interval $[\gamma_i^{\min}, \gamma_i^{\max}]$ for $\{K_i\theta_a\}$ from $[a_i^{\min}, a_i^{\max}]$ and $R_i$ according to Table \ref{table:theta-i-ci} \end{varwidth}
\State $\theta_l = \frac{R_i \cdot \pi/2 + \gamma_i^{\min}}{K_i}$
\State $\theta_u = \frac{R_i \cdot \pi/2 + \gamma_i^{\max}}{K_i}$
\If{$\theta_u-\theta_l < 2\epsilon$}
\State Break
\EndIf
\State $k_{i} = \texttt{FindNextK}(k_{i}, \theta_l, \theta_u)$
\EndWhile
\EndWhile
\State $[a_l, a_u] = [\sin^2(\theta_l), \sin^2(\theta_u)]$
\State \Return $[a_l, a_u]$
\end{algorithmic}
\end{algorithm}

\begin{algorithm}
\caption{FindNextK}
\label{alg:find-next-k}
\begin{algorithmic}[1]
\Require $k_{i}$, $\theta_l$, $\theta_u$
\State $K_{i} = 2k_{i} + 1$
\State $K = \left\lfloor \frac{\pi/2}{\theta_u - \theta_l} \right\rfloor$ 
\If {$K$ is even}
\State $K = K - 1$
\EndIf
\While{$K \geq 3K_{i}$}
\If {$\left\lfloor \frac{K\theta_l}{\pi/2} \right\rfloor = \left\lceil \frac{K\theta_u}{\pi/2} \right\rceil - 1$}
\State \Return $(K - 1)/2$
\EndIf
\State $K = K - 2$
\EndWhile
\State \Return $k_{i}$
\end{algorithmic}
\end{algorithm}
The theorem below provides the precision and query complexity of the MIQAE algorithm, where the query complexity refers to the number of applications of $\mathcal{A}$.
\begin{theorem}\label{thm:algorithm}
Given a confidence level $1-\alpha \in (0, 1)$, a target accuracy $\epsilon > 0$, and an $(n + 1)$-qubit unitary $\mathcal{A}$ satisfying 

\[\mathcal{A}\ket{0}_n\ket{0} = \sqrt{a}\ket{\psi_0}\ket{0} + \sqrt{1-a}\ket{\psi_1}\ket{1},\]
where $\ket{\psi_0}$ and $\ket{\psi_1}$ are  $n$-qubit states and $a \in [0, 1]$,  {the} MIQAE algorithm outputs a confidence interval for $a$ that satisfies
\[\mathbb{P}\left[a \not \in [a_l, a_u]\right] \leq \alpha,\]
where $a_u - a_l < 2\epsilon$, leading to an estimate $\hat{a}=\frac{a_l+a_u}{2}$ for $a$ such that $|a - \hat{a}| < \epsilon$ with a confidence of $1 - \alpha$, using $O\left(\frac{1}{\epsilon}\log \frac{1}{\alpha}\right)$ applications of $\mathcal{A}$. 
\end{theorem}

\section{The  distributed quantum approximate counting algorithm}\label{3}
In this section, we design a distributed  quantum counting algorithm. We first describe  the counting problem.

For any finite set \( X = \{0, 1, \cdots, N-1\} \), which contains several marked elements, let the set of marked elements be denoted as  \( S \). The counting problem is to determine the number of marked elements, i.e., $|S|$. 

Without loss of generality,  assume $N = 2^n$ and $|S|=t$.
For any bit string {of length} $N$, we can
construct a new bit string {of length} $N' = 2^{\lceil \log_2 N\rceil}$ by appending zeros without changing the number of marked elements.  We encode each element $x\in X$ as the computational basis state $\ket{x}$ and prepare the uniform superposition
$
\frac{1}{\sqrt{2^n}}\sum_{x=0}^{2^n-1}\ket{x}.
$
 Let 
\begin{equation}\label{def-A}
   \left( \mathcal{H}^{\otimes n}\otimes \mathcal{I}\right)\ket{0}_{n+1} = \cos\theta\sum_{x\notin S}\frac{\ket{x}_{n}}{\sqrt {2^{n}-t}} \ket{0} + \sin\theta\sum_{x\in S}\frac{\ket{x}_{n}}{\sqrt {t}}\ket{0} ,
\end{equation}
where  $\sin\theta=\sqrt {\frac{t}{2^{n}}}$. The indicator function of the set \( S \) is constructed as follows: $\forall x\in\{0,1\}^n$,
\begin{align}\chi_S(x) =\left\{\begin{aligned}
1 \quad& x\in S,\\
0 \quad& x\notin S.
\end{aligned}\right.\end{align}
There exists an oracle that can identify whether $ x $ is a marked element, that is, 
\begin{equation}\label{01}
U_{\chi_S}\ket{x}_{n}\ket{0}=\ket{x}_{n}\ket{\chi_S(x)}.
\end{equation}
Let $\mathcal{A}=U_{\chi_S}\left(\mathcal{H}^{\otimes n}\otimes \mathcal{I}\right)$. We have \begin{equation}\label{1}
\mathcal{A}\ket{0}_{n}\ket{0}=\cos\theta\sum_{x\notin S}\frac{\ket{x}_{n}}{\sqrt {2^{n}-t}} \ket{0} + \sin\theta\sum_{x\in S}\frac{\ket{x}_{n}}{\sqrt {t}}\ket{1}.
\end{equation}
The task is to estimate $t$ with as few calls to  $\mathcal{A}$
as possible.  By estimating the amplitude $\sin^2\theta$, we can successfully estimate the number of marked elements.

To design a distributed quantum counting algorithm, {we consider
\(2^k\) quantum computing nodes indexed by
\(j\in\{0,1\}^k\), where \(1\leq k<n\). For each node \(j\), let $S_j =\{i'|ji'\in S \}$, where \(i'\) denotes the remaining \(n-k\) bits following the prefix \(j\). 
The \(j\)-th node is assigned the subset \(S_j\)
 and  has access to} an oracle
  $U_{\chi_{S_j}}$ such that\begin{equation}U_{\chi_{S_j}}\ket{i'}_{n-k}\ket{0}=\ket{i'}_{n-k}\ket{\chi_{S_j}(i')},
\end{equation}
where  
$\chi_{S_j}(i')$ is the indicator function of $S_j$. 
Let $|S_j|=t_j$. We define the amplitude amplification operator as \begin{equation}
Q_j = -\mathcal{A}_j\left(\mathcal{I}_{n-k+2}-2 \ket{0}_{n-k+2}\bra{0}_{n-k+2} \right)\mathcal{A}^{\dagger}_j \left(\mathcal{I}_{n-k+2}-2\left(\mathcal{I}_{n-k} \otimes\ket{11}\bra{11} \right)\right), 
\end{equation}
where \begin{equation}
\mathcal{A}_j=\left(\mathcal{I}_{n-k+1}\otimes \mathcal{R}_{r_i}\right)\left(U_{\chi_{S_j}}\otimes \mathcal{I}\right)\left(\mathcal{H}^{\otimes {n-k}}\otimes \mathcal{I}_2\right),~j\in\{0,1\}^k
\end{equation}
\begin{equation}
\mathcal{R}_{r_i}\ket{0}=\sqrt{r_i}\ket{1}+\sqrt{1-r_i}\ket{0}, \quad0<r_i\le1.
\end{equation}
{Here, \(i\) denotes the iteration index. The rotation \(\mathcal{R}_{r_i}\) controls the amplitude of the additional
ancilla qubit. Consequently, the amplitude associated with a marked
element and the outcome \(\lvert 11\rangle\) is multiplied by
\(\sqrt{r_i}\), while the set \(S_j\) itself remains unchanged.}

We denote $$U_{11} = \mathcal{I}_{n-k+2}-2\left(\mathcal{I}_{n-k} \otimes\ket{11}\bra{11} \right),~U_0 = \mathbb{I}_{n-k+2} - 2 \ket{0}_{n-k+2}\bra{0}_{n-k+2}. $$ 
Then $Q_j =-\mathcal{A}_j U_0\mathcal{A}^{\dagger}_j U_{11}$. {Here, \(U_{11}\) reverses the phase of the component in which the last
two qubits are in \(\lvert 11\rangle\), while
\(A_jU_0A_j^\dagger\) performs a reflection about the state prepared by
\(A_j\). Their composition forms the amplitude amplification operator
\(Q_j\).} The circuit for $Q_j$ is shown in Fig. \ref{DIQC algorithm}.

 In the following expression, whenever the summation set is empty, the
corresponding term is understood to be zero. The initial state  of the quantum algorithm is as follows: 
\begin{equation*}\label{def-Aj}
\begin{aligned}
   & \mathcal{A}_j\ket{0}_{n-k+2}  \\= &\sqrt {\frac{r_it_j}{2^{n-k}}}\sum_{i'\in S_j}\frac{\ket{i'}_{n-k}}{\sqrt {t_j}}\ket{11}  + \sqrt {\frac{t_j\left(1-r_i\right)}{2^{n-k}}}\sum_{i'\in S_j}\frac{\ket{i'}_{n-k}}{\sqrt {t_j}} \ket{10}\\
&+\sqrt {\frac{r_i(2^{n-k}-t_j)}{2^{n-k}}}\sum_{i'\notin S_j}\frac{\ket{i'}_{n-k}}{\sqrt {2^{n-k}-t_j}}\ket{01}+\sqrt {\frac{\left(2^{n-k}-t_j\right)\left(1-r_i\right)}{2^{n-k}}}\sum_{i'\notin S_j}\frac{\ket{i'}_{n-k}}{\sqrt {2^{n-k}-t_j}} \ket{00}.
\end{aligned}
\end{equation*}
 We denote 
\begin{equation*}
\begin{aligned}
\ket{\phi_0}=&\sum_{i'\in S_j}\frac{\ket{i'11}}{\sqrt {t_j}},\\\ket{\phi_1}=&\sqrt {\frac{2^{n-k}}{2^{n-k}-t_jr_i}}\left(\sqrt {\frac{1-r_i}{2^{n-k}}}\sum_{i'\in S_j}\ket{i'10} +\sqrt {\frac{r_i}{2^{n-k}}}\sum_{i'\notin S_j}\ket{i'01}+\sqrt {\frac{1-r_i}{2^{n-k}}}\sum_{i'\notin S_j}\ket{i'00}\right).
\end{aligned}
\end{equation*}
 Let  $\sin\widetilde{\theta_j}=\sqrt {\frac{r_it_j}{2^{n-k}}}$.  Then we can obtain 
\begin{equation}\label{a22}
\mathcal{A}_j\ket{0}_{n-k+2}  =\sin\widetilde{\theta_j}\ket{\phi_0}+\cos\widetilde{\theta_j}\ket{\phi_1}.
\end{equation}
When \(t_j=0\) or \(r_i t_j=2^{n-k}\), the corresponding
zero-amplitude term in Eq.~(\ref{a22}) is omitted. Thus, Eq.~(\ref{a22}) remains valid in
these boundary cases.
 {Eq.~(\ref{a22}) separates the state into two orthogonal components.
The state \(\lvert\phi_0\rangle\) corresponds to the desired outcome
\(\lvert 11\rangle\) of the last two qubits, whereas
\(\lvert\phi_1\rangle\) contains all the remaining outcomes. Therefore,
measuring the last two qubits yields \(\lvert 11\rangle\) with probability
$
r_i\frac{t_j}{2^{n-k}}=r_i a_j$, where \(a_j=t_j/2^{n-k}\in[0,1]\) denotes the fraction of marked elements in the \(j\)-th quantum computing node.
The parameter \(r_i\) 
rescales the success probability and  accordingly
adjusts the  angle used in
the subsequent estimation procedure.}
 Let $\sin{\theta}_j=\sqrt {\frac{t_j}{2^{n-k}}}.$ Then $\sin\widetilde{\theta_j}=\sqrt {r_i}\sin{\theta_j}.$ As long as we estimate the amplitude $\sin^2\theta_j$, we can determine $t_j$, and ultimately calculate $t=\sum_{j=0}^{2^k-1}t_j$.

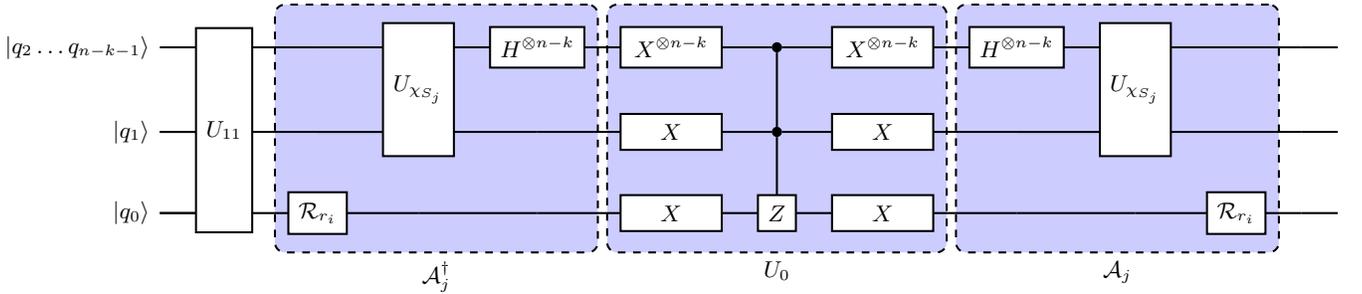
\begin{figure}[h]
		\centering
		\begin{adjustbox}{width=1\textwidth}
\begin{quantikz}
\lstick{$\ket{q_2\ldots q_{n-k-1}}$} &\gate[3]{U_{11}}    & \qw  \gategroup[wires=3,steps=3,style={dashed,rounded
corners,fill=blue!20, inner
xsep=2pt},background,label style={label
position=below,anchor=north,yshift=-0.2cm}]{$\mathcal{A}^{\dagger}_j$}&  \gate[2]{U_{\chi_{S_j}}}& \gate[1]{H^{\otimes n-k}}  &\gate[1][1.4cm]{X^{\otimes n-k}}\gategroup[wires=3,steps=3,style={dashed,rounded
corners,fill=blue!20, inner
xsep=2pt},background,label style={label
position=below,anchor=north,yshift=-0.2cm}]{$U_0$} &\ctrl{1}&\gate[1][1.4cm]{X^{\otimes n-k}}&\gate[1]{H^{\otimes n-k}}
 \gategroup[wires=3,steps=3,style={dashed,rounded
corners,fill=blue!20, inner
xsep=2pt},background,label style={label
position=below,anchor=north,yshift=-0.2cm}]{$\mathcal{A}_j$} &\gate[2]{U_{\chi_{S_j}}} 
& \qw& \qw& \qw\\
\lstick{$\ket{q_1}$}& \qw  & \qw& \qw& \qw&\gate[1][1.4cm]{X} &\ctrl{1}&\gate[1][1.4cm]{X}&\qw&\qw&\qw& \qw& \qw\\
\lstick{$\ket{q_0}$}& \qw  & \gate[1]{\mathcal{R}_{r_i}} & \qw& \qw  
&\gate[1][1.4cm]{X} &\gate[1]{Z}&\gate[1][1.4cm]{X}
& \qw&\qw& \gate[1]{\mathcal{R}_{r_i}}& \qw& \qw
\end{quantikz}
	\end{adjustbox}
\caption{The circuit of the $Q_j$ operator in the DIQC algorithm.}
		\label{DIQC algorithm}
\end{figure}

Below, we present a distributed quantum counting(DIQC) algorithm. $2^k$ computing nodes execute algorithm in parallel. In fact, if there are not two quantum computers, a single quantum computer with $n-k+2$ qubits can still perform the tasks. However, we need to sacrifice some time to let it execute the subsequent algorithms node by node.  A central classical computer assigns {a different
task}  to {each of the} $2^k$ quantum computing nodes, {which execute their
assigned tasks independently.} After each quantum computing node completes the final step of the algorithm, {it sends} the results back to the central computer via classical communication. Finally, the central computer provides an estimate based on the received results.

{For clarity, we briefly explain several parameters and intermediate
variables used in Algorithm~\ref{diqc}. The integer \(K_i\) is the angle-scaling factor in the \(i\)-th
iteration, where \(K_i=2k_i+1\), and \(k_i\) is the number of
applications of \(Q_j\).
The parameter \(r_i\) adjusts the measured amplitude, while \(q\)
controls the minimum growth of \(K_i\). Moreover, \(N_0\)  denotes the prescribed number of shots at each measurement step, and
\(N\) denotes the cumulative number of measurements in the current
iteration.}
 {In the \(i\)-th iteration, let
$
p_{j,i}
=
\sin^2\!\left(K_i\widetilde{\theta}_j\right)
$
denote the true probability of obtaining \(\lvert 11\rangle\), and let
\(\widehat{p}_{j,i}\) denote the corresponding empirical frequency
obtained from the measurement outcomes. The confidence interval
$
[p_{j,i}^{\min},p_{j,i}^{\max}]
$
is constructed from \(\widehat{p}_{j,i}\) to estimate \(p_{j,i}\). Let
\([\widetilde{\theta}_{j,i}^{\min},
\widetilde{\theta}_{j,i}^{\max}]\)
denote the angle interval obtained directly from the confidence interval
centered at the empirical frequency \(\widehat{p}_{j,i}\). Since the measured amplitude is
scaled by \(r_i\), the confidence interval for the original angle
\(\theta_j\) is subsequently recovered from this intermediate interval. The Boolean variable \(\mathtt{backtrack}\) indicates whether
backtracking occurs in the current iteration. It is set to
\(\mathtt{true}\) when the current angle interval and \(K_i\) are restored
to their values from the preceding iteration, and to \(\mathtt{false}\)
otherwise. In Algorithm~\ref{findnext}, this flag determines whether a new candidate
value of \(r_i\) is computed.}

{We next introduce the notation used in the post-processing step. For each iteration \(i\), let
\begin{equation*}
a_{j,i}^{l}=\sin^2\bigl(\theta_{j,i}^{\min}\bigr),
\qquad
a_{j,i}^{u}=\sin^2\bigl(\theta_{j,i}^{\max}\bigr).
\end{equation*}
Let
\begin{equation*}
\mathcal{I}_j=
\left\{
i|
 a_{j,i}^{u}-a_{j,i}^{l}\leq 3\epsilon_j
\right\}
\end{equation*}
denote the set of iterations whose output intervals are retained for
post-processing. For each \(i\in\mathcal{I}_j\), define
\begin{equation*}
a_{j,i}=
\frac{a_{j,i}^{l}+a_{j,i}^{u}}{2},
\qquad
\omega_{j,i}=
\frac{1}{a_{j,i}^{u}-a_{j,i}^{l}}.
\end{equation*}
The weighted amplitude  is defined as
$
c_j=
\frac{
\sum_{i\in\mathcal{I}_j}a_{j,i}\omega_{j,i}
}{
\sum_{i\in\mathcal{I}_j}\omega_{j,i}
}.
$
Each retained interval provides an estimate of \(a_j\). Since a narrower
interval represents a more precise estimate, the inverse interval width
is used as its weight. Consequently, more precise intervals contribute
more strongly to the weighted average. Let
\begin{equation*}c_j^{\min}=\max\{0,c_j-\frac{3\epsilon_j}{2}\},\qquad
c_j^{\max}=\min\{1,c_j+\frac{3\epsilon_j}{2}\}.\end{equation*}}
{\quad At each iteration, the measurement results are first used to construct
a confidence interval for the amplified success probability. This
interval is then converted into an intermediate angle interval for the
adjusted angle. If either endpoint of this intermediate interval satisfies
$
\sin^2\bigl(\widetilde{\theta}_{j,i}\bigr)>r_i,
$
the inverse transformation
\begin{equation*}
\theta_{j,i}
=
\arcsin\sqrt{
\frac{\sin^2\bigl(\widetilde{\theta}_{j,i}\bigr)}{r_i}
}
\end{equation*}
is not well defined. In this case, the current update is discarded, and
the angle interval and \(K_i\) are restored to their values from the
preceding iteration. Otherwise, the effect of \(r_i\) is removed to obtain
an updated confidence interval for the original angle \(\theta_j\).
Algorithm~{\ref{findnext}} then selects the values of \(K_i\) and, when necessary, \(r_i\)
for the next iteration.}
\begin{breakablealgorithm}
\caption{DIQC algorithm}
\label{diqc}
\begin{algorithmic}[1]
\Statex \textbf{Input:} $0<\epsilon\le0.01$, $0<\alpha< \frac{3}{4}$,  the oracle
\(U_{\chi_{S_j}}\),  $N_0\in \mathbb{N}^{+}$ 
\Statex \textbf{Output:}The estimated value of $t'$ 
\State The central computer assigns tasks to the quantum computers, and each quantum computer executes the following steps.
\State $\epsilon_j =\frac{\epsilon}{2^k}$, $\alpha_j =\frac{\alpha}{2^k}$, $i = 0$
\State $[\theta_{j,0}^{min}, \theta_{j,0}^{max}] = [0, \pi/2]$ 
\State $r_0=1$, $K_0=1$ 
\State $K_{max} = 2\left\lfloor {\frac{\pi}{8\epsilon_j}-\frac{1}{2}} \right\rfloor+1$ 
\While{$\sin^2(\theta_{j,i}^{max}) -\sin^2(\theta_{j,i}^{min}) > 2\epsilon_j$}
\If{ $\sin^2(\theta_{j,i}^{max}) -\sin^2(\theta_{j,i}^{min}) \ge 50\epsilon_j$}
\State  $q=2$
\Else 
\State  $q=3$
\EndIf 
\State $N = 0$, $i = i+ 1$
\State $K_{i}=K_{i-1}$, $r_{i}=r_{i-1}$
\State $R_{i} = \lfloor \frac{2K_i \theta_{j,i-1}^{min}}{\pi} \rfloor$ 
\State $\alpha_{j,i} = \frac{(q-1)\alpha_j}{q} \frac{K_i}{K_{max}}$
\State $N_{i}^{max} =  \left\lceil \frac{2}{\sin^2\frac{\pi}{21}\sin^2\frac{8\pi}{21}}\ln\left(\frac{2}{\alpha_{j,i}}\right) \right\rceil$
\While{$K_{i-1}=K_{i}$}
\State Prepare circuit $Q_j^{k_i}\mathcal{A}_j\ket{0}_{n-k+2}$, where $k_i=\frac{K_{i}-1}{2}$
\State Measure $\min\{N_0, N_{i}^{max}-N\}$ times 
\State \begin{varwidth}[t]{.95\linewidth} Combine the results of the current round to approximate ${\widehat{p}}_{j,i} = \mathbb{P}\left[11\right]$ \end{varwidth}
\State $N = N + \min\{N_0, N_{i}^{max}-N\}$
\State $\epsilon_{p_{j,i}} = \sqrt{\frac{1}{2N} \ln\left(\frac{2}{\alpha_{j,i}}\right)}$ 
\State $p_{j,i}^{\max} = \min(1,{\widehat{p}}_{j,i} + \epsilon_{p_{j,i}})$
\State $p_{j,i}^{\min} = \max(0, {\widehat{p}}_{j,i} - \epsilon_{p_{j,i}})$
\If{$R_i$ is even}
\State  $\widetilde{\theta_{j,i}^{min}} = \frac{R_i \cdot \pi/2 +\sin^{-1} \sqrt{p_{j,i}^{\min}}}{K_i}$, $\widetilde{\theta_{j,i}^{max}}= \frac{R_i \cdot \pi/2 + \sin^{-1} \sqrt{p_{j,i}^{\max}}}{K_i}$
\Else
\State $\widetilde{\theta_{j,i}^{min}} = \frac{(R_i+1) \cdot \pi/2 -\sin^{-1} \sqrt{p_{j,i}^{\max}}}{K_i}$,
 $\widetilde{\theta_{j,i}^{max}}= \frac{(R_i+1)\cdot \pi/2 - \sin^{-1} \sqrt{p_{j,i}^{\min}}}{K_i}$
\EndIf 
\If{$\sin^2(\widetilde{\theta_{j,i}^{min}})>r_i$ or $\sin^2(\widetilde{\theta_{j,i}^{max}})>r_i$}
\State {$\theta_{j,i}^{min}=\theta_{j,i-1}^{min}$, $\theta_{j,i}^{max}=\theta_{j,i-1}^{max}$, $K_i=K_{i-1}$, Backtracking occurs }
 \Else
\State $\theta_{j,i}^{min}=\arcsin\sqrt{\frac{\sin^2(\widetilde{\theta_{j,i}^{min}})}{r_i}}$, $\theta_{j,i}^{max}=\arcsin\sqrt{\frac{\sin^2(\widetilde{\theta_{j,i}^{max}})}{r_i}}$
\EndIf 
\If{$\sin^2(\theta_{j,i}^{max}) -\sin^2(\theta_{j,i}^{min}) \le 2\epsilon_j$}
\State Break
\EndIf 
\State FindNext K
\EndWhile
\EndWhile
\State Post-process the retained intervals indexed by
\(\mathcal{I}_j\) and calculate \(c_j\). Choose
\(t'_j\in\mathbb{N}\) such that
$$
|t'_j-2^{n-k}c_j|\leq \frac{2}{3}.
$$
\State \Return $[2^{n-k}c_j^{min}, 2^{n-k}c_j^{max}]$ and $ t'_j\in\mathbb{N}$
\State The central computer receives the results $t'_j$ from all nodes, and then the sum yields $t'$ .
\end{algorithmic}
\end{breakablealgorithm}

{The same-quadrant condition is essential because the function
\(\sin^2 x\) is not one-to-one across different quadrants. If the scaled
angle interval crosses a quadrant boundary, the measured probability
cannot be mapped back to a unique angle interval. Increasing \(K_i\)
magnifies the difference between possible angles and thereby improves
the estimation resolution. When a larger \(K_i\) would cause the scaled
interval to cross a quadrant boundary, \(r_i\) is used to contract the
corresponding angle interval. Algorithm~\ref{findnext} therefore selects \(K_i\) and
\(r_i\) jointly to improve the resolution while preserving an unambiguous
inverse mapping.}

\begin{algorithm}
\caption{FindNext K}
\label{findnext}
\begin{algorithmic}[1]
\Require  $\theta_{j,i}^{min}$, $\theta_{j,i}^{max}$, $q$, $K_{i}$, backtrack
\State $\kappa = \left\lfloor {\frac{\pi}{4(\theta_{j,i}^{max} - \theta_{j,i}^{min})}-\frac{1}{2}} \right\rfloor$ 
\State $K = 2\kappa + 1$
\While{$K\geq qK_{i}$}
\If {$\left\lfloor \frac{2K\theta_{j,i}^{min}}{\pi} \right\rfloor = \left\lceil \frac{2K\theta_{j,i}^{max}}{\pi} \right\rceil - 1$}
\State \Return $r_{i}=1$, $K_{i}=K$
\EndIf
\If { no backtracking occurs}
\State $R_i = \lfloor \frac{2K \theta_{j,i}^{min}}{\pi} \rfloor$ 
\State $r= \frac{\sin^2\left({\frac{(R_i+1)\pi}{2K}}\right)}{\sin^2{(\theta_{j,i}^{max}})}$
\If {$r>\max(\sin^2\left(\frac{\pi}{2} \left(1 - \frac{1}{K}\right)\right),\frac{3}{4})$ and $\left\lfloor \frac{2K\arcsin\left(\sqrt{r}sin{\theta_{j,i}^{min}}\right)}{\pi} \right\rfloor = \left\lceil \frac{2K\arcsin\left(\sqrt{r}sin{\theta_{j,i}^{max}}\right)}{\pi} \right\rceil - 1$}
\State \Return $r_{i}=r$, $K_{i}=K $
\EndIf
\EndIf
\State $K = K - 2$
\EndWhile
\State \Return $K_{i}=K_{i-1}$
\end{algorithmic}
\end{algorithm}

\subsection{ Correctness analysis}
 In this section, we prove the correctness of the DIQC {algorithm}. We  provide  a theorem stating that each quantum computing node can output results with accuracy $\epsilon_j=\frac{\epsilon}{2^k}$ at a confidence level $1-\alpha_j=1-\frac{\alpha}{2^k} $.

  { For convenience, let
$
a_j=\sin^2\theta_j=\frac{t_j}{2^{n-k}}.
$
Thus, \(t_j=2^{n-k}a_j\). Here, \(a_j\) denotes the true amplitude
parameter associated with the \(j\)-th quantum computing node, whereas \(a_{j,i}\) denotes the midpoint
of the \(i\)-th retained interval and is used as an estimate of \(a_j\).} Recall that
 $c_j^{min}=\max(0, c_j-\frac{3\epsilon_j}{2})$,  $c_j^{max} =\min(c_j+\frac{3\epsilon_j}{2},1)$. 
 We have the following result.
 \begin {theorem}\label{r1}
 Given a confidence level $1-\frac{4}{3}\alpha \in (0, 1)$ and target accuracy $\epsilon\in (0, 0.01]$, each quantum computing node  can output  $[2^{n-k}c_j^{min}, 2^{n-k}c_j^{max}]$ satisfying $$\mathbb{P}[2^{n-k}a_j \not \in [2^{n-k}c_j^{min}, 2^{n-k}c_j^{max}]] < \frac{4}{3}\alpha_j,$$
 with $O\left(\frac{1}{\epsilon_j}\ln \frac{1}{\alpha_j}\right)$ queries, where  $2^{n-k}c_j^{max} - 2^{n-k}c_j^{min} \le 2^{n-k}3\epsilon_j$. Additionally, the estimated  value $t'_j\in\mathbb{N}$  satisfies \begin{equation*}
\mathbb{P}\left[\lvert t'_j- 2^{n-k}a_j\rvert\le\frac{2}{3}+2^{n-k-1}3\epsilon_j\right] > 1-\frac{4}{3}\alpha_j.
\end{equation*}

\end {theorem}

The proof of this theorem is divided into two parts: the first part proves that $$\mathbb{P}[2^{n-k}a_j \not \in [2^{n-k}c_j^{min}, 2^{n-k}c_j^{max}]] < \frac{4}{3}\alpha_j,$$ and the second part proves the query complexity of the algorithm. {The} two parts of the proof are presented below.

Notice that  Algorithm \ref{diqc} requires $\theta_{j,i}^{max}$ and $\theta_{j,i}^{min}$ to remain in the same quadrant after each scaling. If {this condition is not satisfied},  we prioritize using $K$ for adjustments, and at {an} appropriate time, incorporate $r$ to ensure faster convergence with a larger $K$. The choice of $K$ and $r$ is aimed at ensuring that the scaled angle remains within the same quadrant.
In fact, this is the purpose of Algorithm \ref{findnext}.

According to Algorithm \ref{findnext}$$\kappa = \left\lfloor \frac{\pi}{4(\theta_{j,i}^{max} - \theta_{j,i}^{min})}-\frac{1}{2}\right\rfloor,$$ and  
$\theta_{j,i}^{max}-\theta_{j,i}^{min}> 2\epsilon_j$, we can set $K_{max}=2\left\lfloor {\frac{\pi}{8\epsilon_j}-\frac{1}{2}} \right\rfloor+1$.   It is clear that  $K_i\le K_{max}$ for any $i$. 

During each round, Algorithm~4 either retains the current value \(K_i\)
or returns a larger odd integer \(K\) satisfying \(K\geq qK_i\). Once
\(N=N_i^{\max}\), the analysis below guarantees that such a larger \(K\)
can be found. When a candidate \(K\geq qK_i\) satisfies
$
\frac{(R_i+1)\pi}{2K}<\theta^{\max},
$
Algorithm~4 attempts to introduce an adjustment factor
$
r\in
\max\left\{
\sin^2\left(\frac{\pi}{2}\left(1-\frac{1}{K}\right)\right),
\frac{3}{4}
\right\}.
$
If the resulting adjusted interval satisfies
$
[\widetilde{\theta}^{\min},\widetilde{\theta}^{\max}]
\subseteq
\left[
\frac{R_i\pi}{2K},
\frac{(R_i+1)\pi}{2K}
\right],
$
the candidate \(K\) can be retained without crossing a quadrant boundary;
otherwise, \(K\) is reduced.

\begin{proposition}\label{proposition21}
 If $1\ge r >\sin^2\left(\frac{\pi}{2} \left(1 - \frac{1}{K}\right)\right)$, $\frac{\pi}{2}> \theta_{j,i}^{min}\ge 0$ and $K$ is an  {integer} greater than or equal to 1,  then $0\le\frac{2K \theta_{j,i}^{min}}{\pi}  -\frac{2K\arcsin\left(\sqrt{r}\sin{\theta_{j,i}^{min}}\right)}{\pi}< 1$.
\end{proposition}
\begin{proof}
Since $\arcsin x$ is an increasing function in the first quadrant and $r \le 1$, it follows that $\arcsin\left(\sqrt{r}\sin{\theta_{j,i}^{min}}\right) \le \theta_{j,i}^{min}$. Therefore, it is evident that $0\le\frac{2K \theta_{j,i}^{min}}{\pi}  -\frac{2K\arcsin\left(\sqrt{r}\sin{\theta_{j,i}^{min}}\right)}{\pi}$. Next, we will prove that $\frac{2K \theta_{j,i}^{min}}{\pi}  -\frac{2K\arcsin\left(\sqrt{r}\sin{\theta_{j,i}^{min}}\right)}{\pi}< 1$. We only need to prove that $\theta_{j,i}^{min} -\arcsin\left(\sqrt{r}\sin{\theta_{j,i}^{min}}\right)< \frac{\pi}{2K } $ for $1\ge r >\sin^2\left(\frac{\pi}{2} \left(1 - \frac{1}{K}\right)\right)$. Let $f(x) = \arcsin\left(\sqrt{r}\sin{x}\right) $, with $x\in[0,\frac{\pi}{2}]$. Then $f^{\prime\prime}(x)<0$, which means $f(x)$ is a concave function. Thus, $x - \arcsin\left(\sqrt{r}\sin{x}\right)$ is a convex function in the first quadrant, and by the properties of convex functions, we know that its maximum value occurs at the endpoints. When $\theta_{j,i}^{min}=0$, it is clear that $\theta_{j,i}^{min} -\arcsin\left(\sqrt{r}\sin{\theta_{j,i}^{min}}\right)< \frac{\pi}{2K } $. When $\theta_{j,i}^{min}=\frac{\pi}{2}$, we have $\frac{\pi}{2}-\arcsin\sqrt{r}<\frac{\pi}{2}-\frac{\pi}{2} \left(1 - \frac{1}{K}\right)=\frac{\pi}{2K }$. Therefore, we  prove $0\le\frac{2K \theta_{j,i}^{min}}{\pi}  -\frac{K\arcsin\left(\sqrt{r}\sin{\theta_{j,i}^{min}}\right)}{\pi/2}< 1$.
\end{proof}

 The result of Proposition \ref{proposition21} is a necessary condition for $\left\lfloor \frac{2K\arcsin\left(\sqrt{r}\sin{\theta_{j,i}^{min}}\right)}{\pi} \right\rfloor\\=\lfloor \frac{2K \theta_{j,i}^{min}}{\pi} \rfloor$. The additional
requirement \(r>\frac{3}{4}\) limits the amount of rescaling and reduces the
likelihood of backtracking. By the definition of \(N_i^{\max}\), $N_{i}^{max}$ measurements are sufficient to {ensure that} $\epsilon_{p_{j,i}}\le\frac{\sin\frac{\pi}{21}\sin\frac{8\pi}{21}}{2}$,  {which allows}  Algorithm \ref{diqc} to proceed to the next round. 
  We represent this conclusion with a proposition.

\begin{proposition}\label{proposition2}
 {For node \(j\) at iteration \(i\), suppose that}  $2\epsilon_{p_{j,i}}\le\sin\frac{\pi}{21}\sin\frac{8\pi}{21}$, $K_i\widetilde{\theta_{j,i}^{min}}$ and $K_i \widetilde{\theta_{j,i}^{max}}$ lie within a single quadrant and $\widetilde{\theta_{j,i}^{min}}\neq\widetilde{\theta_{j,i}^{max}}$. Then  Algorithm \ref{findnext} returns an odd integer  $K\ge3K_i$ and $r\in(\sin^2\left(\frac{\pi}{2} \left(1 - \frac{1}{K}\right)\right),1]$ such that 
$$\left\lfloor \frac{2K\arcsin\left(\sqrt{r}\sin{ \theta^{min}}\right)}{\pi} \right\rfloor = \left\lceil \frac{2K\arcsin\left(\sqrt{r}\sin{ \theta^{max}}\right)}{\pi} \right\rceil - 1.$$
\end{proposition}
\begin{proof}
From lines 23 and 24 of the algorithm, it can be seen that $p_{j,i}^{\max}-p_{j,i}^{\min} \le2\epsilon_{p_{j,i}}$.
Combining lines 26 and 28 {of} the algorithm, we have \begin{align}\label{a10}
p_{j,i}^{\max}-p_{j,i}^{\min}=
 &\left|\sin^2(K_i \widetilde{\theta_{j,i}^{max}})-\sin^2(K_i\widetilde{\theta_{j,i}^{min}})\right|\\
 \le&\sin\frac{\pi}{21}\sin\frac{8\pi}{21}.\end{align} 
 Let $K_i \widetilde{\theta_{j,i}^{max}}=\theta^{max}$, $K_i\widetilde{\theta_{j,i}^{min}}=\theta^{min}$. We consider the case in the first quadrant, and similar results can be obtained for the other quadrants. Let $2\theta={\theta^{max}}+{\theta^{min}}$ and $2\epsilon_\theta={\theta^{max}}-{\theta^{min}}$. Define the following function $$\epsilon_\theta(\theta)=\min \left( \frac{1}{2}\arcsin\left(\frac{\sin\frac{\pi}{21}\sin\frac{8\pi}{21}}{\sin2\theta}\right), \frac{\pi}{2}-\theta, \theta \right),$$ for $\theta\in[0,\frac{\pi}{2}]$. 
  Given an integer $m$, divide the first quadrant into 
$m$ intervals, namely $\left[0, \frac{\pi}{2m}\right],\left[ \frac{\pi}{2m}, \frac{2\pi}{2m}\right], \cdots, \left[ \frac{(m-1)\pi}{2m}, \frac{\pi}{2}\right]$, and define the piecewise  linear function on each interval: 
\begin{equation}
f_m(x)=\min \left( x- \frac{(p-1)\pi}{2m},  \frac{p\pi}{2m}-x \right), x\in\left[ \frac{(p-1)\pi}{2m}, \frac{p\pi}{2m}\right], p\in\{1,\cdots,m\}.
\end{equation}
Observe that as long as the current $\left[\theta^{min},\theta^{max}\right]$ is within these intervals, multiplying by $m$ will still keep it within the same quadrant.  We consider the case of $m\ge3$. Let $f_{\max}(x)= \max\left(f_3(x), f_5(x), f_7(x)\right)$.
 According to Appendix A of Ref.~\cite {Fukuzawa2023}, it is known that $\epsilon_\theta\le f_{\max}(\theta)$ for $\theta\in[0,\frac{\pi}{2}]$. Therefore, there exists an \(m\in\{3,5,7\}\) such that multiplying the
interval by \(m\) keeps it within a single quadrant. Taking \(K=mK_i\)
and \(r=1\), we obtain an odd integer \(K\geq3K_i\) such that \begin{equation}\left\lfloor \frac{2K\arcsin\left(\sqrt{r}\sin{  \theta^{min}}\right)}{\pi} \right\rfloor = \left\lceil \frac{2K\arcsin\left(\sqrt{r}\sin{  \theta^{max}}\right)}{\pi} \right\rceil - 1.\end{equation}
\end{proof}

{The number of measurements is increased gradually up to
\(N_i^{\max}\). When \(N=N_i^{\max}\), the error bound required in
Proposition \ref{proposition2} is guaranteed. In practice, Algorithm \ref{findnext} may find a valid
value \(K'\) earlier, with \(K'\geq qK_i\), where
\(q\in\{2,3\}\) is the amplification factor used in Algorithm \ref{diqc}.}

{Assume that} $T-1$ is  the smallest positive integer that satisfies $\sin^2(\theta_{j,i}^{max}) -\sin^2(\theta_{j,i}^{min}) > 2\epsilon_j$. Below, we present a lemma concerning the relationship {among the values} $K_i(0\le i\le T)$.

\begin {lemma} \label{l1}
Let \(K_1,\ldots,K_T\) be a sequence satisfying
\(K_T\leq K_{\max}\) and $K_i\ge q K_{i-1}>1$
for $1\le i \le T$, where $2\le q\le3$. If  the function $f$ is  increasing on $[K_1, K_{max}]$,  then \begin{equation}\label{increasing-sum}
    \sum_{i=\hat{i}}^{T} f(K_i) \leq \sum_{i=0}^{T-\hat{i}} f\left(\frac{K_{max}}{q^{i}}\right),
\end{equation}where $\hat{i}\in\{1,\cdots,T\}$ {be an arbitrary starting iteration.}
\end{lemma}
\begin{proof}
 Since \(K_{\ell}\geq qK_{\ell-1}\), we have
\[
K_i
\leq
\frac{K_{i+1}}{q}
\leq\cdots\leq
\frac{K_t}{q^{T-i}}
\leq
\frac{K_{\max}}{q^{T-i}}.
\]
Since \(f\) is increasing,
\[
f(K_i)
\leq
f\left(\frac{K_{\max}}{q^{T-i}}\right).
\]
Therefore,
\[
\sum_{i=\hat{i}}^{T}f(K_i)
\leq
\sum_{i=\hat{i}}^{T}
f\left(\frac{K_{\max}}{q^{T-i}}\right)
=
\sum_{i=0}^{T-\hat{i}}
f\left(\frac{K_{\max}}{q^i}\right).
\]
\end{proof}

{Recall that
\begin{equation*}
\mathcal I_j=\{i\mid  a_{j,i}^{u}-a_{j,i}^{l}\leq 3\epsilon_j\}
\end{equation*}
denotes the set of iterations retained for post-processing. By the definition of \(c_j\), the
corresponding counting estimate before rounding is
\begin{equation*}
2^{n-k}c_j
=
2^{n-k}
\frac{\sum_{i\in\mathcal I_j}a_{j,i}\omega_i}
{\sum_{i\in\mathcal I_j}\omega_i},
\end{equation*} which
is not necessarily an integer.}

We will first prove $$\mathbb{P}[2^{n-k}a_j \not \in [2^{n-k}c_j^{min}, 2^{n-k}c_j^{max}]] < \frac{4}{3}\alpha_j,$$ where $c_j^{max}=\min(c_j+\frac{3\epsilon_j}{2},1)$, $c_j^{min}=\max(0, c_j-\frac{3\epsilon_j}{2})$, $2^{n-k}c_j^{max} - 2^{n-k}c_j^{min} \le 2^{n-k}3\epsilon_j$. Thus, we can further prove that \begin{equation*}
\mathbb{P}\left[\lvert t'_j- 2^{n-k}a_j\rvert\le\frac{2}{3}+2^{n-k-1}3\epsilon_j\right] > 1-\frac{4}{3}\alpha_j,
\end{equation*}
where $t'_j$ is the estimated value.
\begin{proof}
Since \(c_j\) is a weighted average of 
\(a_{j,i}\) with positive weights, it lies between their minimum and
maximum values. Therefore,
\[
\min_{i\in\mathcal I_j} a_{j,i}
\leq c_j\leq
\max_{i\in\mathcal I_j} a_{j,i}.
\]
 Let $a_j=\sin^2{\theta}_j$.   If, { for every} $i\in
\mathcal{I}_j$, $a_j \in [a_{j,i}^{l}, a_{j,i}^{u}]$,   $a_{j,i}^{u} - a_{j,i}^{l} \le 3\epsilon_j$, { and} $a_{j,i}$ is the midpoint of $a_{j,i}^{u}$ and $a_{j,i}^{l}$, { then} we have $|a_{j,i}-a_j|\le\frac{3\epsilon_j}{2}$. Therefore, we have $\min\limits_{i\in\mathcal I_j} a_{j,i}\geq a_j-\frac{3\epsilon_j}{2}$, $\max\limits_{i\in\mathcal I_j} a_{j,i}\leq a_j+\frac{3\epsilon_j}{2}$.  Ultimately we can obtain $|c_j-a_j|\le\frac{3\epsilon_j}{2}$. Since $c_j^{max}=\min(c_j+\frac{3\epsilon_j}{2},1)$, $c_j^{min}=\max(0, c_j-\frac{3\epsilon_j}{2})$,
we have $a_j\in[c_j^{min}, c_j^{max}]$ and $2^{n-k}c_j^{max} - 2^{n-k}c_j^{min} \le 2^{n-k}3\epsilon_j$.

As $2^{n-k}$, $c_j^{min}$, and $c_j^{max}$ are all non-negative numbers, it follows that \begin{equation}
 \mathbb{P}[2^{n-k}a_j\not \in [2^{n-k}c_j^{min}, 2^{n-k}c_j^{max}]]=  \mathbb{P}[a_j\not \in [c_j^{min}, c_j^{max}]].
\end{equation}
{If all iteration-wise confidence intervals contain the corresponding
true probabilities \(p_{j,i}\) of obtaining \(\lvert 11\rangle\),
then the final interval contains the true value \(a_j\). Therefore, failure of the final interval implies
that at least one iteration produces a confidence interval that does not
contain the corresponding true probability. By the union
bound, the overall failure probability can then be bounded by the sum
of the failure probabilities over all iterations.}
We have  the following equation:
\begin{align*}
\mathbb{P}[2^{n-k}a_j \not \in [2^{n-k}c_j^{min}, 2^{n-k}c_j^{max}]]
=&\mathbb{P}[a_j\not \in [c_j^{min}, c_j^{max}]]\\
   \le &\mathbb{P}[\exists i \in\{1,\cdots,T\}, p_{j,i}\not \in [p_{j,i}^{min}, p_{j,i}^{max}]]\\
      \le & \sum_{i=1}^{T}\mathbb{P}[p_{j,i}\not \in [p_{j,i}^{min}, p_{j,i}^{max}]]\\
     \le & \sum_{i=1}^{T}2e^{-2N{\epsilon_{p_{ji}}}^2} \\
     =& \sum_{i=1}^{T}{\alpha_{j,i}}\\
     =&  \sum_{i=1}^{T}\frac{(q-1)\alpha_j}{q} \frac{K_i}{K_{max}}.
\end{align*}
 The first inequality is due to the fact that if {the event} $p_{j,i} \in [p_{j,i}^{min}, p_{j,i}^{max}]$ occurs for all $1\le i \le T$, then combined {this event} with $R_i$ and $r_i$, we can infer $\theta_{j}\in[\theta_{j,i}^{min},\theta_{j,i}^{max}]$, thus $a_j \in [a_{j,i}^{l}, a_{j,i}^{u}]$, so  $a_j  \in [c_j^{min}, c_j^{max}]$ can be deduced in the end. 
 The third inequality arises because we choose the Chernoff-Hoeffding bound~\cite {Hoeffding1994} to construct the confidence interval, requiring $\epsilon_{p_{j,i}}= \sqrt{\frac{1}{2N} \ln\left(\frac{2}{\alpha_{j,i}}\right)}$.

{Applying} Lemma \ref{l1} {with} $f(x)=x$, 
 we obtain \begin{equation}\label{a4}
     \sum_{i=\hat{i}}^{T}K_i \leq \sum_{i=0}^{T-\hat{i}} \frac{K_{max}}{q^{i}}.
\end{equation}
{Let $q_i\in\{2,3\}$ denote the growth factor selected in the
\(i\)-th iteration. Then
$
K_i\geq q_iK_{i-1}.
$} 
Since
$
K_i\geq q_iK_{i-1}\geq 2K_{i-1}
$
and $K_T\leq K_{\max}$, we have
\[
\frac{K_i}{K_{\max}}
\leq
\frac{1}{2^{T-i}}.
\]
Moreover,
\[
\frac{q_i-1}{q_i}\leq \frac{2}{3}.
\]
Therefore,
\begin{align}
\sum_{i=1}^{T}\alpha_{j,i}
&=
\sum_{i=1}^{T}
\frac{(q_i-1)\alpha_j}{q_i}
\frac{K_i}{K_{\max}}
\\
&\leq
\frac{2\alpha_j}{3}
\sum_{i=1}^{T}\frac{1}{2^{T-i}}
=
\frac{2\alpha_j}{3}
\sum_{i=0}^{T-1}\frac{1}{2^i}
<
\frac{4}{3}\alpha_j.
\end{align}

Thus, {this} proves that \begin{equation}\label{a6}
\mathbb{P}[2^{n-k}a_j \not \in [2^{n-k}c_j^{min}, 2^{n-k}c_j^{max}]] <\frac{4}{3}\alpha_j.\end{equation}
According to the post-processing step of  Algorithm \ref{diqc}, the output value $t'_j$ must be an integer and satisfy $$\lvert t'_j- 2^{n-k}c_j\rvert\le\frac{2}{3}.$$ Furthermore, since we have previously proven $|c_j-a_j|\le\frac{3\epsilon_j}{2}$, it follows that \begin{align}
\lvert t'_j- 2^{n-k}a_j\rvert=&\lvert t'_j- 2^{n-k}c_j+2^{n-k}c_j- 2^{n-k}a_j\rvert\\
\le&\lvert t'_j- 2^{n-k}c_j\rvert+\lvert 2^{n-k}c_j- 2^{n-k}a_j\rvert\\
\le&\frac{2}{3}+2^{n-k-1}3\epsilon_j.\label{eq:36}
\end{align}
Therefore, {combining Eqs.~\eqref{a6} and \eqref{eq:36}, we obtain}
\begin{equation}\label{a7}
\mathbb{P}\left[\lvert t'_j- 2^{n-k}a_j\rvert\le\frac{2}{3}+2^{n-k-1}3\epsilon_j\right] > 1-\frac{4}{3}\alpha_j.
\end{equation}
\end{proof}

Next, we prove that the {query complexity} for each quantum computing node is $O\left(\frac{1}{\epsilon_j}\ln \frac{1}{\alpha_j}\right)$.
\begin{proof}
{In the \(i\)-th iteration, each measurement requires
$\frac{K_i-1}{2}$ applications of \(Q_j\), and at most
\(N_i^{\max}\) measurements are performed.  Hence, the total query
complexity is bounded by the sum of the query costs over all
iterations, as follows:}  \begin{equation}\label{a9}\sum_{i=1}^{T}k_iN_{i}\le\sum_{i=1}^{T}k_iN_{i}^{max}=\sum_{i=1}^{T}\frac{K_i-1}{2}N_{i}^{max}\le\sum_{i=1}^{T}\frac{K_i}{2}N_{i}^{max}.\end{equation}
Substituting the value of $N_{i}^{max}$ into $\sum_{i=1}^{T}\frac{K_i}{2}N_{i}^{max}$, we obtain
 \begin{equation}\label{a10}\sum_{i=1}^{T}\frac{K_i}{2}N_{i}^{max}=\sum_{i=1}^{T}\frac{K_i}{\sin^2\frac{\pi}{21}\sin^2\frac{8\pi}{21}}\ln\left(\frac{2qK_{max}}{{(q-1)K_i \alpha_j }}\right)=\sum_{i=1}^{T}c K_i\ln\left(\frac{2qK_{max}}{{(q-1)K_i \alpha_j }}\right),\end{equation} where  $$c=\frac{1}{\sin^2\frac{\pi}{21}\sin^2\frac{8\pi}{21}}.$$
 Since \(q\in \{2,3\}\) in each iteration, we have
$
\frac{2q}{q-1}\le 4.$ Therefore,
\[
\sum_{i=1}^{T} cK_i
\ln\left(
\frac{2qK_{\max}}
{(q-1)K_i\alpha_j}
\right)
\le
\sum_{i=1}^{T} cK_i
\ln\left(
\frac{4K_{\max}}
{K_i\alpha_j}
\right).
\] 
Let $$f(x)=x\ln\left(\frac{4K_{max}}{{x \alpha_j }}\right).$$ Then  $$f^{\prime}(x)=\ln\left(\frac{4K_{max}}{{x \alpha_j }}\right)-1.$$For \(x\in[K_1,K_{\max}]\), since \(x\le K_{\max}\) and
\(\alpha_j<1\), we have
\[
\frac{4K_{\max}}{x\alpha_j}
\ge \frac{4}{\alpha_j}
>4>e.
\]
Hence, \(f'(x)>0\), and thus \(f(x)\) is increasing on
\([K_1,K_{\max}]\).  Since 
$K_i\geq qK_{i-1}$, we have
$
K_i\geq 2K_{i-1}.
$
Therefore, applying Lemma~\ref{l1} with \(q=2\) yields \begin{align}\label{a11}
\sum_{i=1}^{T}c K_i\ln\left(\frac{4K_{max}}{{K_i \alpha_j }}\right)
\le&\sum_{i=0}^{T-1}\frac{cK_{max}}{2^{i}}\ln\left(\frac{2\cdot2^{i+1}}{{ \alpha_j }}\right)\nonumber\\
=&cK_{max}\left[\ln 2\sum_{i=0}^{T-1}\frac{i}{2^{i}}+\sum_{i=0}^{T-1}\frac{1}{2^{i}}\ln\left(\frac{4}{{\alpha_j }}\right)\right]\nonumber\\
<& {cK_{max}}\left[2\ln 2+2\ln\frac{4}{{ \alpha_j }}\right]\nonumber\\
=& \left(2c\left\lfloor {\frac{\pi}{8\epsilon_j}-\frac{1}{2}} \right\rfloor+c\right)\left[3\ln 4+2\ln\frac{1}{{ \alpha_j }}\right]\\
&\in O\left(\frac{1}{\epsilon_j}\ln \frac{1}{\alpha_j}\right).\nonumber\end{align}
{Here, we  use
 \begin{equation*}
\sum_{i=0}^{\infty}\frac{1}{2^i}=2,
\qquad
\sum_{i=0}^{\infty}\frac{i}{2^i}=2.
\end{equation*}}
\end{proof}The following is a theorem that addresses the counting problem by aggregating the results from various nodes.
\begin {theorem} \label{r2}
  {The} DIQC algorithm can output   the estimated value 
$t'=\sum_{j=0}^{2^k-1}t'_j$ satisfying $$\mathbb{P}\left[|t'-t| \leq 2^{n-k-1}3\epsilon+\frac{2^{k+1}}{3}\right]>1-\frac{4}{3}\alpha.$$
\end {theorem}
\begin{proof}
For each \(j\in\{0,1,\cdots,2^k-1\}\), let
$
t_j=2^{n-k}a_j,
$
and let \(t'_j\) denote the estimate returned by the \(j\)-th node.
Then
\[
t=\sum_{j=0}^{2^k-1}t_j,
\qquad
t'=\sum_{j=0}^{2^k-1}t'_j.
\]
For each node \(j\), define the event
\[
E_j=
\left\{
|t'_j-t_j|
\le
\frac{2}{3}+2^{n-k-1}3\epsilon_j
\right\}.
\]
According to Theorem \ref{r1},
$
\Pr(E_j)>1-\frac{4}{3}\alpha_j,
$
and hence
$
\Pr(E_j^c)<\frac{4}{3}\alpha_j.
$
By the union bound,
\[
\begin{aligned}
\Pr\left(\bigcap_{j=0}^{2^k-1}E_j\right)
&=
1-\Pr\left(\bigcup_{j=0}^{2^k-1}E_j^c\right)\\
&\ge
1-\sum_{j=0}^{2^k-1}\Pr(E_j^c)\\
&>
1-\frac{4}{3}\sum_{j=0}^{2^k-1}\alpha_j\\
&=
1-\frac{4}{3}\alpha,
\end{aligned}
\]
where \(\alpha_j=\alpha/2^k\).

When all the events \(E_j\) occur, we have
\[
\begin{aligned}
|t'-t|
&=
\left|
\sum_{j=0}^{2^k-1}(t'_j-t_j)
\right|\\
&\le
\sum_{j=0}^{2^k-1}|t'_j-t_j|\\
&\le
2^k\left(
\frac{2}{3}+2^{n-k-1}3\epsilon_j
\right)\\
&=
\frac{2^{k+1}}{3}
+
2^{n-k-1}3\epsilon,
\end{aligned}
\]
where \(\epsilon_j=\epsilon/2^k\). Therefore,
\[
\Pr\left[
|t'-t|
\le
2^{n-k-1}3\epsilon+\frac{2^{k+1}}{3}
\right]
>
1-\frac{4}{3}\alpha.
\]
\end{proof}
{There is a trade-off between the number of quantum computing nodes and the resources required by each node. For a given value of \(k\), the algorithm uses \(2^k\) quantum computing nodes, and each node requires \(n-k+2\) qubits in total. Therefore, increasing \(k\) by one doubles the number of nodes while reducing the number of qubits required by each node by one. However, since
\[
\epsilon_j=\frac{\epsilon}{2^k},
\qquad
\alpha_j=\frac{\alpha}{2^k},
\]
the accuracy and confidence requirements assigned to each node become more stringent. According to Theorem~\ref{r1}, the query complexity of each node is
\[
O\left(
\frac{1}{\epsilon_j}
\ln\frac{1}{\alpha_j}
\right)
=
O\left(
\frac{2^k}{\epsilon}
\ln\frac{2^k}{\alpha}
\right).
\]
 Therefore, increasing \(k\) reduces the number of qubits used by each node, but may increase the number of queries required by each node. The choice of \(k\) should thus balance the number of quantum computing nodes, the number of qubits required by each node, the query complexity.}
\subsection{ Comparison with related works}
Due to the relationship between our algorithm and  {the} MIQAE algorithm in Ref.~\cite {Fukuzawa2023}, we will first compare them. Although Ref.~\cite {Fukuzawa2023} does not mention counting, the current counting algorithm is actually closely related to the amplitude estimation algorithm.
Considering the issue of circuit depth, we adjust the process of finding $K$ by introducing a parameter that allows the algorithm to reach the target width more quickly, while ensuring that the maximum $K_{max}$ is less than or equal to the $K_{max}$ in Ref.~\cite {Fukuzawa2023}. To ensure a high success probability in the end, we do not take the last interval that reached the target width as the output value of the algorithm. Instead, we appropriately relax the conditions  and  {retain all amplitude intervals whose widths do not exceed 1.5 times the target width as candidate intervals}. The weighted average of the midpoints of these candidate values is then taken as  the amplitude estimated value. {The corresponding counting estimate is obtained by multiplying the amplitude estimate by the size of the search space.} Throughout the entire algorithm process, we divide our amplification factors into two stages: first 2, then 3. This approach is taken in order to reduce the number of measurements when the circuit depth is large. In most numerical examples, the algorithm performs better when using two amplification factors than when using a single amplification factor.  If K is still not found, the algorithm in Ref.~\cite {Fukuzawa2023} may become stuck and fail to terminate. We account for this situation in our Qiskit implementation. Specifically, once all allowed measurement attempts are exhausted, the implementation outputs the current confidence interval and terminates.  
 From this perspective, our algorithm is more robust. The numerical experiments presented later also demonstrate that our algorithm performs better in terms of circuit depth and success rate.

We apply the algorithm from Ref.~\cite {brassard2000quantum} to find the number of marked elements and compare it with our algorithm.  Define the  operator $Q$ as follows: $Q=-\mathcal{A} U_0\mathcal{A}^{\dagger} U$, where  $$\mathcal{A}=\mathcal{H}^{\otimes n}, U\ket{x}_{n}=(-1)^{\chi_S(x)}\ket{x}_{n}, U_0 = \mathbb{I}_{n} - 2 \ket{0}_{n}\bra{0}_{n}. $$The algorithm from Ref.~\cite {brassard2000quantum} is illustrated in Fig.\ref{QC algorithm}. By using the conclusion from Ref.~\cite {mottonen2004quantum}, when a multi-qubit gate acts on an $n$-qubit state, the gate can be decomposed, requiring $4^n$ single-qubit gates and $4^n-2^{n+1}$ controlled-NOT gates. Therefore, one application of the controlled-$Q$ operator requires $4^{n+1}+4^{n+1}-2^{n+2}$. Implementing all applications of the controlled-$Q$ operator in the circuit shown in Fig.\ref{QC algorithm} requires $(2^m-1)(4^{n+1}+4^{n+1}-2^{n+2})$ quantum gates. The  implementation of  $\mathcal{A}$ requires $n$ quantum gates, and the implementation of ${QFT}^{\dagger}_{2^m}$ requires $\frac{m^{2}+m}{2}$ quantum gates. Therefore, the total number of
quantum gates in the circuit is
$
(2^m-1)(4^{n+1}+4^{n+1}-2^{n+2})
+n+\frac{m^2+m}{2}+m.
$

In Algorithm \ref{diqc}, one application of $Q_j$ requires  $4^{n-k+2}+4^{n-k+2}-2^{n-k+3}$ quantum gates. Let $k_{max}=\left\lfloor {\frac{\pi}{8\epsilon_j}-\frac{1}{2}} \right\rfloor$. The maximum number of applications of \(Q_j\) in a single circuit of
Algorithm \ref{diqc} does not exceed \(k_{\max}\).  Therefore, the number of quantum gates required to implement  $Q_j^{k_{max}}$  is $k_{max}(2^{2n-2k+5}-2^{n-k+3})$. Implementing {an} $\mathcal{A}_j$ operator requires $4^{n-k+2}+4^{n-k+2}-2^{n-k+3}$ quantum gates. Therefore, the total number of quantum gates in the circuit
$Q_j^{k_{\max}}\mathcal{A}_j$ is
$
(k_{\max}+1)
\left(2^{2n-2k+5}-2^{n-k+3}\right).
$
Consider obtaining an estimated result with a probability of $\frac{8}{\pi^2}$ , and the estimation error is $\frac{1}{2}$. Let $m=n+1$, $\epsilon_j=\frac{1}{3\cdot2^{n}}$.
For the quantum counting
algorithm in Ref.~\cite {brassard2000quantum}, it requires a total of \(2^{n+1}-1\) applications of the controlled-\(Q\) operator.
Since
\begin{equation}\label{a99}
k_{\max}
\le
3\pi 2^{n-3}-\frac{1}{2}
<
2^{n+1}-1,
\end{equation}
 the maximum number of applications of the Grover operator in a single
DIQC circuit is smaller. 
Moreover, since \(k\geq1\),
\begin{equation}\label{a100}
4^{n-k+2}+4^{n-k+2}-2^{n-k+3}
\le
4^{n+1}+4^{n+1}-2^{n+2}.
\end{equation}
Combining Ineqs.~\eqref{a99} and \eqref{a100}, we obtain
 $$
\frac{n^{2}+7n+4}{2}+2^{n+2}  (4^{n+1}-2^{n+2}+1)>(2^{2n-2k+5}-2^{n-k+3})(3\cdot2^{n-3}\pi+\frac{1}{2}). $$
 Therefore, the maximum gate count of a single DIQC circuit is also
smaller. Table \ref{tabel1} is used to illustrate the comparison between the two algorithms. 
 
 We further compare the circuit depths of the two algorithms. The operator \(Q_j\) acts on \(n-k+2\) qubits, whereas the
controlled-\(Q\) operator acts on \(n+1\) qubits. Since \(k\geq 1\),
we have \(n-k+2\leq n+1\). Let \(d(\cdot)\) denote the circuit depth. Under the same gate decomposition scheme, we have
\[
d(\mathcal{A}_j)\leq d(Q_j)
\quad\text{and}\quad
d(Q_j)\leq d(\text{controlled-}Q).
\] 
Moreover,
$
k_{\max}+1
\le
3\pi2^{n-3}+\frac{1}{2}
<
2^{n+1}-1.
$
Hence,
the depth of a single DIQC circuit satisfies
\[
d(\mathcal{A}_j)+k_{\max}d(Q_j)
<
(2^{n+1}-1)d(\text{controlled-}Q).
\]
Since the quantum counting circuit additionally contains the initial
Hadamard gates, \(\mathcal{A}\), and the inverse QFT, the 
circuit depth of DIQC algorithm is smaller.
 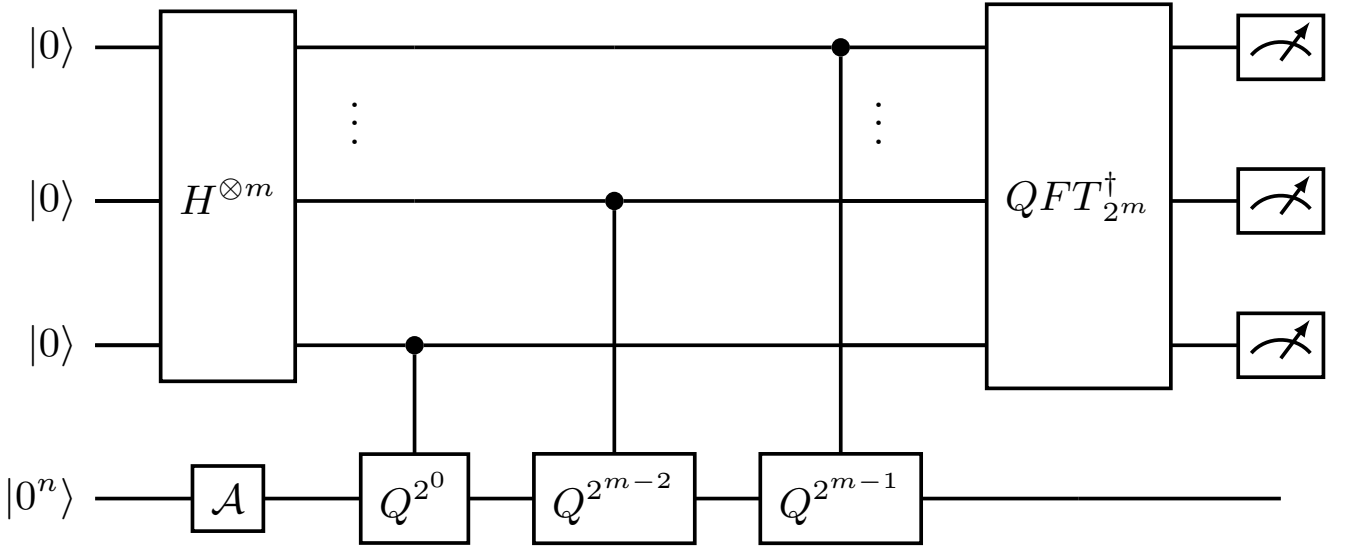
\begin{figure}[h]
		\centering
		\begin{adjustbox}{width=1\textwidth}
\begin{quantikz}
\lstick{$\ket{0}$} &\gate[3]{{H}^{\otimes m}}    & \qw & \qw  &\ctrl{3}&\gate[3]{{QFT}^{\dagger}_{2^m}}& \meter{}\\
\lstick{$\ket{0}$}& \qw   & \qw& \ctrl{2} & \qw& \qw& \meter{}\\
\lstick{$\ket{0}$}& \qw  & \ctrl{1} & \qw& \qw  & \qw&\meter{}\\
\lstick{$\ket{0^{n}}$}& \gate[1]{{\mathcal{A}}}  & \gate[1]{Q^{2^0}} & \gate[1]{Q^{2^{m-2}}} & \gate[1]{Q^{2^{m-1}}}   
& \qw& \qw
\end{quantikz}
\begin{tikzpicture}[overlay, remember picture]
  \node at (-7.5,1.25) {$\vdots$};
   \node at (-3.5,1.25) {$\vdots$};
\end{tikzpicture}
	\end{adjustbox}
\caption{The circuit  {for the}  quantum counting algorithm in Ref.~\cite {brassard2000quantum}.}
		\label{QC algorithm}
\end{figure}

\begin{table}[h] 
	\centering
	\caption{Comparison of our algorithm with {the} quantum counting  algorithm in Ref.~\cite {brassard2000quantum}.}
	\footnotesize
	\begin{tabular}{*{4}{c}}
		\toprule
		          Algorithms      & \makecell {Number \\of qubits} &\makecell{ The number of\\ quantum gates} &\makecell{Total number of\\ applications  of\\ the Grover operator} \\
		\hline
		 \makecell{ Counting\\ algorithm~\cite {brassard2000quantum} }         &      $2n+1$ &\quad$\frac{n^{2}+7n+4}{2}+2^{n+2}  (4^{n+1}-2^{n+2}+1)$   &$2^{n+1}-1$    \\		 \makecell  {Our \\algorithm\ref{diqc} }      &$n-k+2$ & {\quad$< (2^{2n-2k+5}-2^{n-k+3})(3\cdot2^{n-3}\pi+\frac{1}{2})$}    &$<3\cdot2^{n-3}\pi-\frac{1}{2}$\\
		\toprule 
	\end{tabular}\label{tabel1}
\end{table}
\section{Application}\label{4}
We apply the DIQC algorithm to estimate the inner product and Hamming distance of  two bit strings, which can be used to solve problems in machine learning, such as {least-squares} fitting~\cite {Hao2023}. 

Consider the inner product problem: Alice  and Bob have $2^n$-dimensional vectors  $x, y\in\{0,1\}^{2^n}$. The task is to compute the inner product $x\cdot y=\frac{1}{2^n}\sum_{i=0}^{2^n-1}x_ {i}y_{i}$ within an error $\epsilon\le0.01$. Assume that  Alice and Bob each have  {access to} their own quantum oracles.  Alice's oracle is $U_x: \ket{i}_n\ket{0}\rightarrow\ket{i}_n\ket{x_ {i}}$, querying data $x_ {i}$,
 and Bob's oracle is $U_y: \ket{i}_n\ket{0}\ket{0}\rightarrow\ket{i}_n\ket{0}\ket{y_ {i}}$, querying data $y_ {i}$. To apply {the} DIQC algorithm, it is essential to encode the problem into an initial state in the form of Eq. (\ref{a22}).
 
 Given \(n\) and \(k\) with \(1\leq k<n\), we decompose the \(n\)-bit index into an \((n-k)\)-bit part and a \(k\)-bit part.
 Accordingly, we decompose the oracles of Alice and Bob as follows: \begin{equation}\label{a23}
U_{x,j'}: \ket{i'}_{n-k}\ket{0}\rightarrow\ket{i'}_{n-k}\ket{x_ {g_{j' }(i')}}, \quad U_{y,j'}: \ket{i'}_{n-k}\ket{0}\ket{0}\rightarrow\ket{i'}_{n-k}\ket{0}\ket{y_ {g_{j' }(i')}},\end{equation} where $g_{j' }(i')=2^{k}i'+j'$ for any $i'\in\{0,1,\cdots,2^{n-k}-1\}$ and $j'\in\{0, 1,\cdots,2^{k}-1\}$. For each fixed $j'$, the $j'$-th Alice--Bob pair processes the entries
indexed by $g_{j'}(i')$ as $i'$ ranges over
$\{0,1,\ldots,2^{n-k}-1\}$. As $j'$ ranges from $0$ to $2^k-1$, all
entries of the original vectors $x$ and $y$ are covered exactly once. 
{Alice and Bob cooperatively implement $\mathcal{A}_{j'}$ through quantum communication.} To implement $ \mathcal{A}_{j'}$,  Alice {first} applies the Hadamard gate to {prepare} an $n-k$-qubit uniform superposition state, and then she applies $U_{x,j'}$ to obtain $$U_{x,j'}\frac{1}{\sqrt{2^{n-k}}} \sum_{i'=0}^{2^{n-k}-1}\ket{i'}_{n-k}\ket{0}=\frac{1}{\sqrt{2^{n-k}}} \sum_{i'=0}^{2^{n-k}-1}\ket{i'}_{n-k}\ket{x_ {g_{j' }(i')}}.$$ Alice sends $n-k+1$ qubits to Bob.  After Bob receives the quantum state, he adds a register and applies a unitary operator $(U_{y,j'}\otimes \mathcal{I})(\mathcal{I}_{n-k}\otimes CCNOT)(U_{y,j'}\otimes \mathcal{I})$ to obtain \begin{align*}
&(U_{y,j'}\otimes \mathcal{I})(\mathcal{I}_{n-k}\otimes CCNOT)(U_{y,j'}\otimes \mathcal{I})\left[\frac{1}{\sqrt{2^{n-k}}} \sum_{i'=0}^{2^{n-k}-1}\ket{i'}_{n-k}\ket{x_ {g_{j' }(i')}}\ket{0}\ket{0}\right]\\
=&(U_{y,j'}\otimes \mathcal{I})(I_{n-k}\otimes CCNOT)\frac{1}{\sqrt{2^{n-k}}} \sum_{i'=0}^{2^{n-k}-1}\ket{i'}_{n-k}\ket{x_ {g_{j' }(i')}}\ket{y_ {g_{j'}(i')}}\ket{0}\\
=&(U_{y,j'}\otimes \mathcal{I})\frac{1}{\sqrt{2^{n-k}}} \sum_{i'=0}^{2^{n-k}-1}\ket{i'}_{n-k}\ket{x_ {g_{j' }(i')}}\ket{y_ {g_{j'}(i')}}\ket{x_ {g_{j' }(i')}\land y_ {g_{j'}(i')}}\\
=&\frac{1}{\sqrt{2^{n-k}}} \sum_{i'=0}^{2^{n-k}-1}\ket{i'}_{n-k}\ket{x_ {g_{j' }(i')}}\ket{0}\ket{x_ {g_{j' }(i')}\land y_ {g_{j'}(i')}}.
\end{align*}
After decoupling, we discard the third register's qubit $\ket{0}$. Then Bob sends $$\frac{1}{\sqrt{2^{n-k}}} \sum_{i'=0}^{2^{n-k}-1}\ket{i'}_{n-k}\ket{x_ {g_{j' }(i')}}\ket{x_ {g_{j' }(i')}\land y_ {g_{j'}(i')}}$$ to Alice, and the subsequent operations are {performed}  by Alice.  The number of qubits required for communication after executing one instance of $\mathcal{A}_{j'}$ is $2n-2k+3$. Alice first performs a decoupling operation by applying $U_{x,j'}\otimes \mathcal{I}$ to obtain  
\begin{equation}\label{a101}
\begin{aligned}
&U_{x,j'}\otimes \mathcal{I}\left[\frac{1}{\sqrt{2^{n-k}}} \sum_{i'=0}^{2^{n-k}-1}\ket{i'}_{n-k}\ket{x_ {g_{j' }(i')}}\ket{x_ {g_{j' }(i')}\land y_ {g_{j'}(i')}}\right]\\=&\frac{1}{\sqrt{2^{n-k}}} \sum_{i'=0}^{2^{n-k}-1}\ket{i'}_{n-k}\ket{0}\ket{x_ {g_{j' }(i')}\land y_ {g_{j'}(i')}}. 
\end{aligned}
\end{equation}
Let
\[
M_{j'}=
\left\{
i' \mid
x_{g_{j'}(i')}y_{g_{j'}(i')}=1
\right\},
\qquad
c_{j'}=|M_{j'}|.
\]
After applying
\(\mathcal{I}_{n-k}\otimes R_{r_i}\otimes\mathcal{I}\)
to the state in Eq.~(\ref{a101}), the probability associated with the outcome
\(\lvert 11\rangle\) of the last two qubits is
$
\frac{r_i c_{j'}}{2^{n-k}}.
$
 Let  $\sin\widetilde{\beta_{j'}}=\sqrt {\frac{r_ic_{j'}}{2^{n-k}}}.$ 
Let \(\mathcal{A}_{j'}\) denote the overall unitary transformation
implemented by the above sequence of operations.  By grouping the component associated with \(\lvert 11\rangle\) and
its orthogonal complement, the resulting state can be written as
\begin{equation}\label{a21}\mathcal{A}_{j'}\ket{0}_{n-k+2}=\sin\widetilde{\beta_{j'}}\ket{\phi_0}+\cos\widetilde{\beta_{j'}}\ket{\phi_1},\end{equation}
where \(\lvert\phi_0\rangle\) and \(\lvert\phi_1\rangle\) are
normalized orthogonal states, with \(\lvert\phi_0\rangle\) associated
with the outcome \(\lvert 11\rangle\) of the last two qubits.
 Accordingly, 
the Grover operator $Q_ {j'}$ is defined as follows:
\begin{equation*}
Q_ {j'}= -\mathcal{A}_{j'}\left(\mathcal{I}_{n-k+2}-2 \ket{0}_{n-k+2}\bra{0}_{n-k+2} \right)\mathcal{A}^{\dagger}_{j'} \left(\mathcal{I}_{n-k+2}-2\left(\mathcal{I}_{n-k} \otimes\ket{11}\bra{11} \right)\right),
\end{equation*}
where $ j\in\{0,1,\cdots,2^{k}-1\}$. The initial state  is {given by} Eq. (\ref{a21}).

 Given $0<\epsilon'\le0.01$ and $0<\alpha'<\frac{3}{4}$,  we can set $\epsilon'_j=\frac{\epsilon'}{2^k}$, $\alpha'_j=\frac{\alpha'}{2^k}$. Let
$
a_{j'}=\frac{c_{j'}}{2^{n-k}}.
$ Without considering rounding, the $2^k$ quantum computers   run  the DIQC algorithm and output $c'_{j'}\in R$ with {probability greater than} $1- \frac{4}{3}\alpha'_j$, satisfying $\lvert c'_{j'}- c_{j'}\rvert\le2^{n-k-1}3\epsilon'_j$. 
The $2^k$ results are then transmitted to a classical computer via classical communication.  {By} Theorem \ref{r2}, we have the following result. 
\begin {corollary}\label{h1}
 Given a confidence level $1-\frac{4}{3}\alpha'\in (0, 1)$ and target accuracy $\epsilon'\in (0, 0.01]$,  an estimate $c=\frac{\sum_{j'=0}^{2^k-1}c'_{j'}}{2^n}$ of $x\cdot y$  can be obtained with  
  \begin{equation*}
\mathbb{P}\left[|c- x\cdot y| \leq 2^{-k-1}3\epsilon'\right]>1-\frac{4}{3}\alpha', 
\end{equation*}
 where  $c'_{j'}$  is the output of each quantum computing node without considering rounding.
\end {corollary}
\begin {proof}
{By the same aggregation argument as in Theorem~\ref{r2}, but without the rounding term,} we have
 \begin{equation*}
\mathbb{P}\left[|2^n\cdot c- 2^n x\cdot y| \leq 2^{n-k-1}3\epsilon'\right]>1-\frac{4}{3}\alpha'. 
\end{equation*}
Since $2^n$ is a positive integer, $|2^n\cdot c- 2^n x\cdot y| \leq 2^{n-k-1}3\epsilon'$ is equivalent to $| c-  x\cdot y| \leq 2^{-k-1}3\epsilon'$. Therefore, {the result follows}.
\end {proof}
We denote $$M=\left(2c\left\lfloor {\frac{\pi}{8\epsilon'_j}-\frac{1}{2}} \right\rfloor+c\right)\left[3\ln 4+2\ln\frac{1}{{ \alpha'_j }}\right],$$
where  $c=\frac{1}{\sin^2\frac{\pi}{21}\sin^2\frac{8\pi}{21}}$.
For each quantum computer, the number of queries does not exceed
\(8M\), and the communication complexity does not exceed
$
(4n-4k+6)M.
$
Therefore, for the \(2^k\) quantum computers, the total communication
complexity does not exceed
\[
2^k(4n-4k+6)M
\in
O\left(
\frac{n}{\epsilon'_{j'}}
\ln\frac{1}{\alpha'_{j'}}
\right).
\]

 The algorithm in Ref.~\cite {Hao2023} {achieves an estimation error} of $\frac{1}{2^t}$ for  $x\cdot y$ with a probability of $\frac{8}{\pi^2}$. Let $t = n + 1$.  We can obtain $\epsilon'=\frac{1}{2^{n + 1}}$. To achieve the same {accuracy and success probability}, we only need to set $\alpha' = \frac{3}{4}-\frac{6}{\pi^2}$ and $\epsilon'=\frac{1}{3\cdot2^{n -k }}$. From Corollary \ref{h1}, we can see that our estimated error reaches $\frac{1}{2^{n + 1}}$. Under the same gate decomposition scheme, \(Q_{j'}\) acts on
\(n-k+2\) qubits, whereas the controlled-\(Q\) operator in
Ref.~\cite {Hao2023} acts on \(n+2\) qubits. Hence,
$
d(Q_{j'})\leq d(\text{controlled-}Q).
$
  Since $k_{max}=\left\lfloor {\frac{\pi}{8\epsilon'_j}-\frac{1}{2}} \right\rfloor$ and $\epsilon'_j=\frac{1}{3\cdot2^{n }}$, we have $k_{max}\le {{3\cdot2^{n -3 }\pi}-\frac{1}{2}}$. Therefore, the circuit depth of our algorithm is $$d(\mathcal{A}_{j'})+k_{max}d(Q_{j'})<\left({{3\cdot2^{n -3 }\pi}+\frac{1}{2}}\right)d(Q_{j'}).$$
 {The initial Hadamard layer has circuit depth \(1\). Therefore,} the circuit depth of the algorithm in Ref.~\cite {Hao2023} is $$1+(2^{n+1}-1)d(\text{controlled-}Q)+d(QFT_{2^{n+1}}^{\dagger}).$$
  With an estimation error of $\frac{1}{{2^{n+1}}}$, we present the following remark.

\begin {remark}
For a single quantum computer, compared to the quantum algorithm for computing inner products in Ref.~\cite {Hao2023}, our algorithm reduces the number of qubits by $n + k-1$, decreases the circuit depth by at least $\left(2^{n +1 }-3\cdot2^{n -3 }\pi-\frac{3}{2}\right)d(Q_{j'}) +d(QFT_{2^{n+1}}^{\dagger})$, but  {incurs an additional logarithmic factor} $O(\ln \frac{1}{\alpha'_j})$ in the communication complexity.
\end {remark}

Next, we consider using {the} DIQC algorithm to calculate the Hamming distance between two vectors. 

For any two given vectors $x=x_0\cdots x_{2^n-1}\in\{0,1\}^{2^n}$, $y=y_0\cdots y_{2^n-1}\in\{0,1\}^{2^n}$, the Hamming distance is defined as the number of positions at which the corresponding bits are different, i.e., $d=\left| \{i|x_i\neq y_i, 0\le i\le 2^n-1\}\right|$. 
We continue to use the unitary {operators} $U_{x,j'}$ and $U_{y,j'}$ defined in Eq.(\ref{a23}). Let $$A_{j'}=(\mathcal{I}_{n-k}\otimes\mathcal{R}_{r_i})U_{y,j'}(\mathcal{I}_{n-k}\otimes CNOT)U_{y,j'}(U_{x,j'}\otimes \mathcal{I})( \mathcal{H}^{\otimes n-k}\otimes \mathcal{I}_2),$$ where $CNOT$ indicates that the rightmost qubit is  the control qubit and the second rightmost qubit is its target qubit. Then we have
\begin{align}
&(\mathcal{I}_{n-k}\otimes R_{r_i})
U_{y,j'}
(\mathcal{I}_{n-k}\otimes CNOT)
U_{y,j'}
(U_{x,j'}\otimes\mathcal{I})
\left(
\frac{1}{\sqrt{2^{n-k}}}
\sum_{i'=0}^{2^{n-k}-1}
|i'\rangle_{n-k}|00\rangle
\right)
\nonumber\\
&=
(\mathcal{I}_{n-k}\otimes R_{r_i})
\frac{1}{\sqrt{2^{n-k}}}
\sum_{i'=0}^{2^{n-k}-1}
|i'\rangle_{n-k}
|x_{g_{j'}(i')}\oplus y_{g_{j'}(i')}\rangle
|0\rangle .
\end{align}

Let
\[
H_{j'}=
\left\{
i'\mid
x_{g_{j'}(i')}\oplus y_{g_{j'}(i')}=1
\right\},
\qquad
b_{j'}=|H_{j'}|.
\]
After applying \(R_{r_i}\), the probability associated with the
outcome \(|11\rangle\) of the last two qubits is
$
\frac{r_i b_{j'}}{2^{n-k}}.
$
Let
$
\sin\widetilde{\beta_{j'}}
=
\sqrt{\frac{r_i b_{j'}}{2^{n-k}}}.
$
Then
\[
\mathcal{A}_{j'}|0\rangle^{n-k+2}
=
\sin\widetilde{\beta_{j'}}|\phi_0\rangle
+
\cos\widetilde{\beta_{j'}}|\phi_1\rangle,
\]
where \(|\phi_0\rangle\) and \(|\phi_1\rangle\) are normalized
orthogonal states, and \(|\phi_0\rangle\) corresponds to the outcome
\(|11\rangle\) of the last two qubits. Accordingly, the corresponding Grover operator is defined as
\[
Q_{j'}
=
-\mathcal{A}_{j'}
\left(
\mathcal{I}_{n-k+2}
-2|0\rangle^{n-k+2}\langle 0|^{n-k+2}
\right)
\mathcal{A}_{j'}^{\dagger}
\left(
\mathcal{I}_{n-k+2}
-2\left(
\mathcal{I}_{n-k}\otimes |11\rangle\langle 11|
\right)
\right).
\]

Similarly, for each $j'\in\{0,1,\cdots,2^{k}-1\}$, we need to utilize quantum communication to implement $A_{j'}$.  Alice {first} needs to send $$\frac{1}{\sqrt{2^{n-k}}} \sum_{i'=0}^{2^{n-k}-1}\ket{i'}_{n-k}\ket{x_ {g_{j' }(i')}}$$ to Bob,  {after which}  the remaining operations are {performed locally} by Bob. Therefore, implementing one instance of $A_{j'}$ requires $n-k+1$ qubits of communication. Each computer requires at most $n-k+2$ qubits.

Similarly, given $0<\epsilon''\le0.01$ and $0<\alpha''<\frac{3}{4}$,  we can set $\epsilon''_j=\frac{\epsilon''}{2^k}$, $\alpha''_j=\frac{\alpha''}{2^k}$. {If an integer output is not required, the rounding step can be omitted.
Then, by the same aggregation argument as in Theorem~\ref{r2}, we obtain the
following result.}
 \begin {corollary}\label{h2}
 Given a confidence level $1-\frac{4}{3}\alpha''\in (0, 1)$ and target accuracy $\epsilon''\in (0, 0.01]$,  an estimate $b=\frac{\sum_{j'=0}^{2^k-1}b'_{j'}}{2^n}$ of $\frac{1}{2^n} \sum_{i} x_i \oplus y_i
$  can be obtained with  
  \begin{equation*}
\mathbb{P}\left[|b- \frac{d}{2^n}| \leq 2^{-k-1}3\epsilon'\right]>1-\frac{4}{3}\alpha'', 
\end{equation*}
 where  $b'_{j'}$  is the output of each quantum computing node without considering rounding.
\end {corollary}
 The total communication complexity required to compute $\frac{d}{2^n}$ with probability $1 - \frac{4}{3}\alpha''$  using $2^k$ quantum computers running {the DIQC} algorithm in parallel does not exceed \begin{equation} 2^k\left(2n-2k+2\right)M'\in O\left(\frac{n}{\epsilon''_j}\ln \frac{1}{\alpha''_j}\right),\end{equation}where $$M'=\left(2c\left\lfloor {\frac{\pi}{8\epsilon''_j}-\frac{1}{2}} \right\rfloor+c\right)\left[3\ln 4+2\ln\frac{1}{{ \alpha''_j }}\right],$$ $c=\frac{1}{\sin^2\frac{\pi}{21}\sin^2\frac{8\pi}{21}}$.  The algorithm in Ref.~\cite {Hao2023} sets the estimation error {to} $\frac{1}{\sqrt{2^{n}}}$, which  {requires}  $t = \frac{n}{2}$ {estimation}  qubits. {Accordingly, the maximum number of qubits required by a single quantum
computer in their algorithm is
$
n+t+1=\frac{3n}{2}+1.$} To achieve an estimation error of $\frac{1}{\sqrt{2^{n}}}$, we set $\epsilon''=\frac{1}{3\cdot\sqrt{2^{n-2k-2}}}$. Then
$
\epsilon''_{j'}
=
\frac{\epsilon''}{2^k}
=
\frac{1}{3\sqrt{2^{n-2}}},
$
and hence the maximum Grover exponent satisfies
$
k_{\max}
\le
3\pi\sqrt{2^{n-8}}-\frac{1}{2}.
$ 
 Furthermore, the maximum circuit depth of our algorithm is $d(\mathcal{A}_{j'})
+k_{\max}d(Q_{j'})$. {The initial Hadamard layer has circuit depth \(1\).} Thus, under the
same gate decomposition scheme, the circuit depth of the algorithm in Ref.~\cite {Hao2023} is $$1+(2^{\frac{n}{2}}-1)d(\text{controlled-}Q)+d(QFT_{2^{\frac{n}{2}}}^{\dagger}).$$With an estimation error of $\frac{1}{\sqrt{2^{n}}}$, we present the following conclusion.
 \begin {remark}
For a single quantum computer, compared to the quantum algorithm for computing Hamming distance in Ref.~\cite {Hao2023}, our algorithm reduces the number of qubits by $\frac{n}{2} + k-1$, decreases the circuit depth by at least $\left(\sqrt{2^n}-3\cdot\sqrt{{2^{n -8 }}}\pi-\frac{3}{2}\right)d(Q_{j'}) +d(QFT_{2^{\frac{n}{2}}}^{\dagger})$, but  {incurs an additional logarithmic factor} $O(\ln \frac{1}{\alpha''_j})$ in the communication complexity. 
\end {remark}
\section{Numerical {experiments}}\label{experiment1}
In this section, we conduct simulation experiments using IBM's quantum programming software,  {Qiskit}. Given the set  $X=\{0, 1, \cdots, 2^6-1\}$, where $S=\{38, 8, 16\}$ {is the set of} marked elements, we  now use the DIQC algorithm to find the number of marked elements. \begin{figure*}[htbp]
		\centering
		\includegraphics[width=0.5\linewidth]{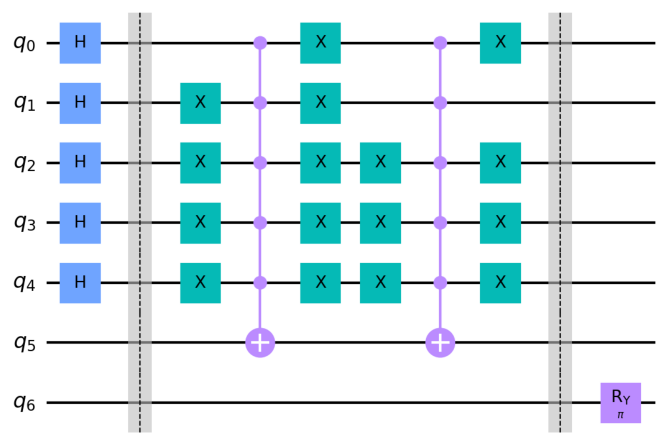}
		\setlength{\abovecaptionskip}{-0.01cm}
		\caption{The circuit for  $\mathcal{A}_0$ in {the} DIQC algorithm.}
		\label{A0}
	\end{figure*} Let $k=1$. Then
$S_0=\{10000, 01000\}$, $S_1=\{00110\}$. Construct $\mathcal{A}_0$ and $\mathcal{A}_1$ as follows:  \begin{equation}
\mathcal{A}_0=\left(\mathcal{I}_{6}\otimes \mathcal{R}_{r_i}\right)\left(U_{\chi_{S_0}}\otimes \mathcal{I}\right)\left(\mathcal{H}^{\otimes {5}}\otimes \mathcal{I}_2\right),
\end{equation}
 \begin{equation}
\mathcal{A}_1=\left(\mathcal{I}_{6}\otimes \mathcal{R}_{r_i}\right)\left(U_{\chi_{S_1}}\otimes \mathcal{I}\right)\left(\mathcal{H}^{\otimes {5}}\otimes \mathcal{I}_2\right),
\end{equation}
where \begin{equation}U_{\chi_{S_j}}\ket{x}_{5}\ket{0}=\ket{x}_{5}\ket{\chi_{S_j}(x)}, x\in\{0,1\}^5, j\in\{0,1\},
\end{equation}
$\chi_{S_j}(x)$ is the indicator function of $S_j$. The {circuits}  for $\mathcal{A}_0$ and $\mathcal{A}_1$ are shown in Fig.\ref{A0} and Fig.\ref{A1}. When implementing the $\mathcal{R}_{r_i}$ operator, we use the $R_Y(\theta)$ gate, where $\theta=2  \arcsin(\sqrt{r_i})$. The {circuits}  provided below correspond to $\mathcal{A}_0$ and $\mathcal{A}_1$  at the initial value $r_i=1$.
As the algorithm runs, the parameter $\theta$ of the $R_Y(\theta)$ gate also changes gradually. 

\begin{figure*}[htbp]
		\centering
		\includegraphics[width=0.5\linewidth]{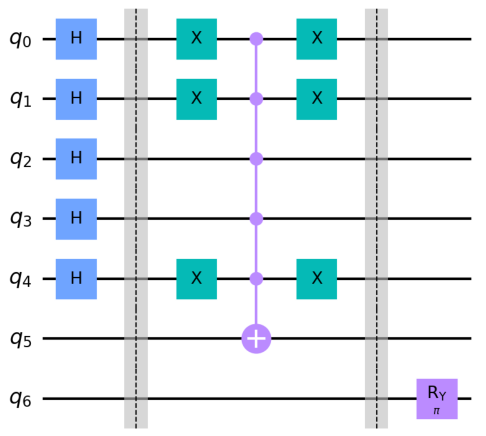}
		\setlength{\abovecaptionskip}{-0.01cm}
		\caption{The circuit for  $\mathcal{A}_1$ in {the} DIQC algorithm.}
		\label{A1}
	\end{figure*}
The operators $Q_j$(for $j\in\{0,1\}$,$r_i=1$)  are defined as follows:
\begin{equation}
Q_j= -\mathcal{A}_j\left(\mathcal{I}_{7}-2 \ket{0}_{7}\bra{0}_{7} \right)\mathcal{A}^{\dagger}_j \left(\mathcal{I}_{7}-2\left(\mathcal{I}_{5} \otimes\ket{11}\bra{11} \right)\right), ~j\in\{0,1\}.
\end{equation}

Given $\alpha_j=0.05$, $\epsilon_j =0.001$, $N_0=1$, 
 we execute the algorithm 100 times for {each of} $j = 0$ and $j = 1$, {and take the average of the successful estimates as the estimated result.}
{All numerical experiments are performed using Qiskit. The algorithm
was executed on the Aer simulator. For each parameter setting, the algorithm is run 100 times.
For each run in which the result is successfully estimated, we
record the largest value of \(k_i\) used during that run. The average
of these largest \(k_i\) values is then calculated, and its ceiling is
taken as the value of \(k_i\) for circuit-resource evaluation. The
corresponding complete circuit \(Q_j^{k_i}A_j\) is transpiled using the
configuration of the 16-qubit IBM Guadalupe backend provided by
Qiskit's fake backend. The transpilation is performed with optimization
level~3 and a fixed transpiler seed of~42. The depth and gate counts of the transpiled circuit are then recorded.
The estimation results and the main resource requirements are presented
in Tables \ref{tabel4} and \ref{tabel5}, respectively.}
  
 \begin{table}[h] 
	\centering
	\caption{{Estimation results of the quantum computing nodes for
\(\alpha_j=0.05\) and \(\epsilon_j=0.001\).}}
	\begin{tabular}{*{6}{c}}
		\toprule
		          \makecell{Computation\\ node}     & \makecell {Number \\of qubits}   &\makecell{Average estimated \\results} &\makecell{Number\\ of successes}\\
		\hline
		 {$j = 0$ }         & 7    &2.0010 &100 \\		  
		 {$j = 1$ }      &7   &1.0003 &100\\
		\toprule 
	\end{tabular}\label{tabel4}
\end{table}

{For node \(j=0\), the transpiled circuit contains $103780$ \(R_Z\) gates,
$39576$ \(SX\) gates, and $5954$ \(X\) gates. For node \(j=1\), the
transpiled circuit contains $56850$ \(R_Z\) gates, $22191$ \(SX\) gates, and
$3577$ \(X\) gates. Here, \(SX\) denotes the square-root of the Pauli-\(X\) gate, satisfying
\(SX^2=X\).}  By summarizing the results of these two computation nodes, we can conclude that the total number of marked elements is 3. {The Qiskit simulations verify the correctness of the
estimated results.} 
\begin{table}[t]
{
\centering
\caption{{Resource requirements and transpilation results of the
quantum computing nodes for \(\alpha_j=0.05\) and
\(\epsilon_j=0.001\)}}
\begin{tabular}{c c c c c c c}
\hline
\makecell{Computation\\ node} &
\makecell{Average\\ query complexity}  &
\makecell{Average\\ maximum \(k_i\)}  &
\makecell{Number of \\ two-qubit gates} &
\makecell{transpiled \\circuit depth}&\\
\hline
\(j=0\) &  13603& 109 & 122983 & 189046  \\
\(j=1\) & 10900 & 89 & 67791 & 103068  \\
\hline
\end{tabular}\label{tabel5}}
\end{table}

 Consider the  performance of {the algorithm} on a single node. We compare our algorithm with the algorithm (MIQAE) in Ref.~\cite {Fukuzawa2023}. Although the MIQAE algorithm is an amplitude estimation algorithm, it can be used for quantum counting. When the number of marked elements divided by the total number of elements is equal to the square of the modulus of the amplitude to be estimated, multiplying the estimation result by the total count gives the number of marked elements. Therefore, without considering the rounding of the estimation result, we take the estimation result from Ref.~\cite {Fukuzawa2023} and multiply it by the total count to obtain the final estimation result. Given that the total count is 64, with one marked element, we conduct 100 experiments on the algorithm {for each parameter setting}.  {We mainly compare the number of successful estimations and the average
maximum Grover power \(k_i\) used during the successful runs. Since the
depth of the complete circuit \(Q^{k_i}\mathcal{A}\) increases with \(k_i\), this
quantity reflects the circuit-depth trend. For a fair comparison, the same values of \(\epsilon\) and \(\alpha\)
are used for both algorithms. Each parameter setting is repeated 100
times. For both algorithms, an estimation is regarded as successful \begin{equation*}
|\widehat{t}-t|\leq 1.5N\epsilon,
\end{equation*}
where \(t\) is the true number of marked elements, \(\widehat{t}\) is
the estimated number of marked elements, and \(N=64\) is the total number
of elements.} The experimental results are shown in Fig.\ref{Q10} and Fig.\ref{Q11}. From the figures, it can be seen that our algorithm has a lower circuit depth and a higher success rate. In practice, we find that the circuit depth is much smaller than $k_{max}$.
\begin{figure*}[htbp]
		\centering
		\includegraphics[width=0.8\linewidth]{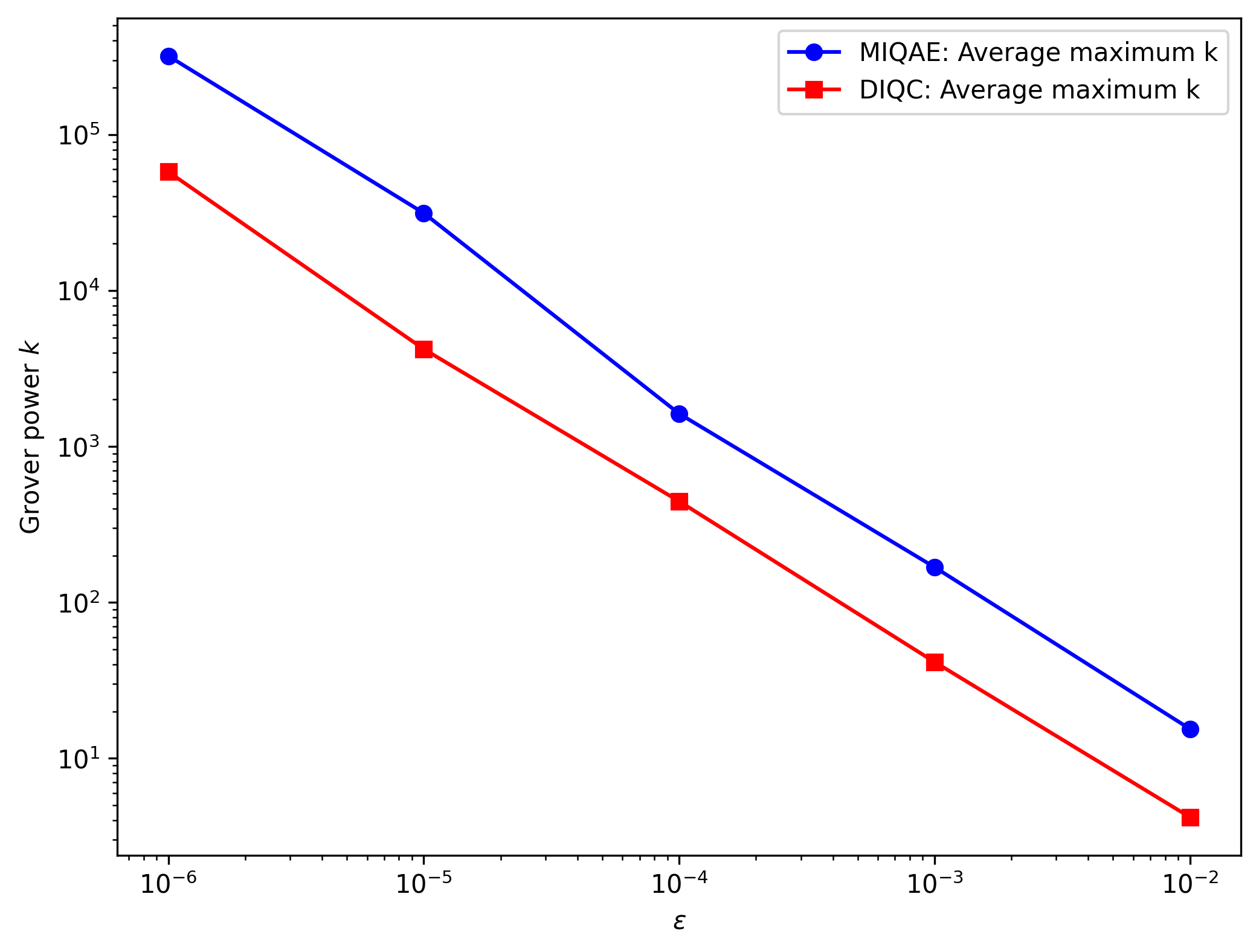}
		\setlength{\abovecaptionskip}{-0.01cm}
		\caption{The relationship between the average maximum  {Grover power} under successful estimation and $\epsilon$ when running the algorithm 100 times with an input amplitude of $1/64$ and $\alpha=0.05$.}
		\label{Q10}
	\end{figure*}
\begin{figure*}[htbp]
		\centering
		\includegraphics[width=0.8\linewidth]{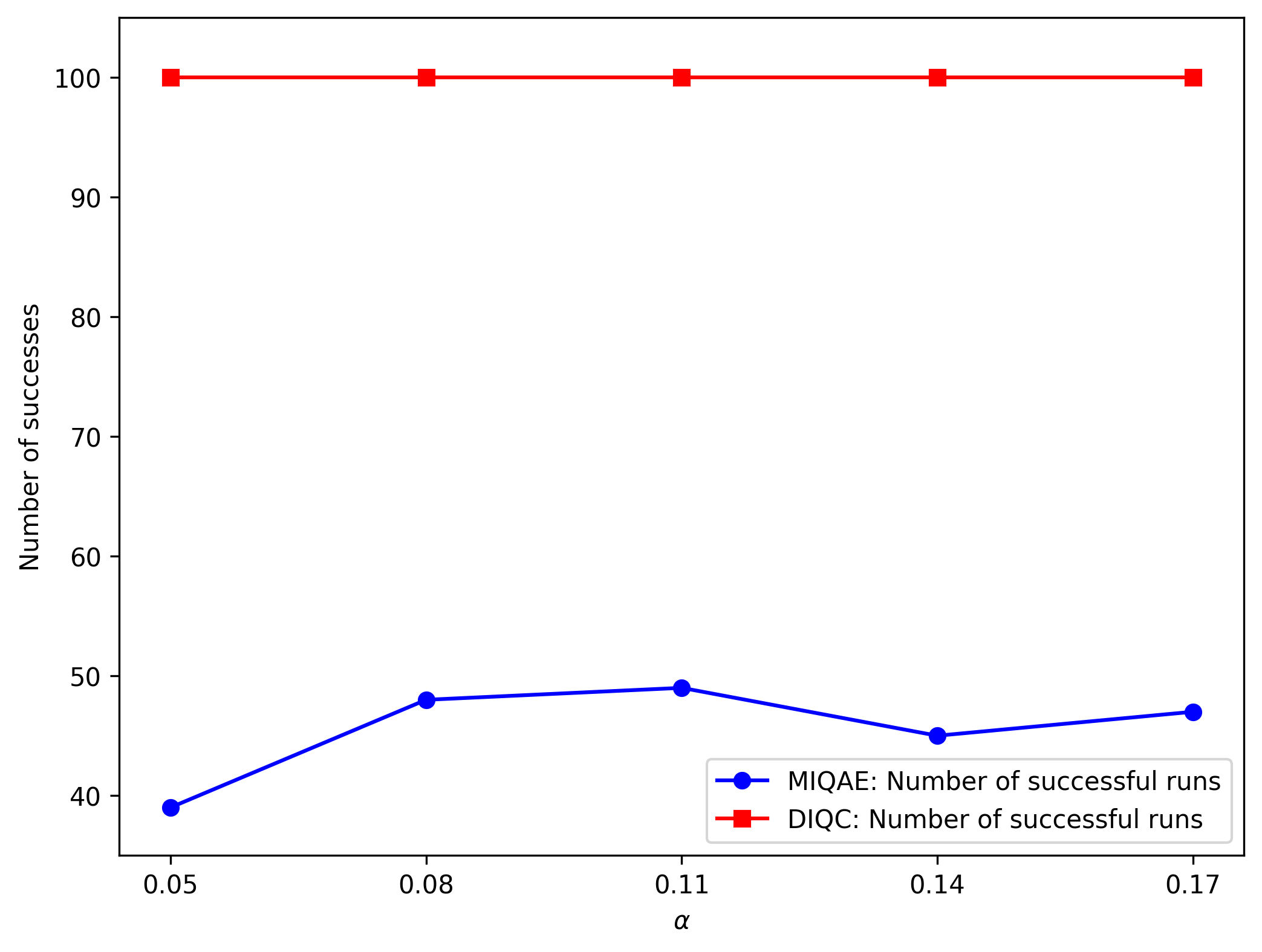}
		\setlength{\abovecaptionskip}{-0.01cm}
		\caption{The relationship between the number of successful estimates and significance level $\alpha$ when running the algorithm 100 times with an input amplitude of $1/64$ and $\epsilon=0.001$.}
		\label{Q11}
	\end{figure*}

\section{Conclusion}\label{6}
  We have proposed a distributed quantum counting algorithm that operates without the need for {the quantum} Fourier transform or controlled Grover operators, relying solely on {the} Grover operator. The algorithm can be executed in parallel or sequentially. We have proven the correctness of the algorithm and analyzed its query complexity, which is related to the confidence level and the interval width. By using a smaller significance level in the early stages, the algorithm can reduce the number of measurements as the circuit depth increases. Compared to quantum counting algorithms with controlled Grover operators, the proposed algorithm reduces both the number of qubits and the number of quantum gates, {and} decreases the circuit depth. We have further demonstrated {through specific examples that our algorithm requires fewer qubits and quantum gates and achieves a lower circuit depth.}  We have applied the proposed algorithm to inner product estimation and Hamming distance calculation. Simulation experiments have been conducted on the Qiskit platform, and the results have further validated the effectiveness  of the algorithm. We have also compared the DIQC algorithm with the MIQAE algorithm, and the experiments have {shown} that {the} DIQC algorithm has lower circuit depth and a higher success probability. The algorithm may {be better suited for} practical application of quantum computing.  Future research could focus on deriving tighter upper bounds for the query complexity and identifying strategies to further minimize the circuit depth. In addition, it may be worthwhile to consider how to design a dynamic adjustment of the amplification factor in the future to better optimize the algorithm.

\section*{Acknowledgements}
The authors are thankful  to the anonymous referees and editor for their significant comments and suggestions that greatly help us improve the quality
of the manuscript.  This work is supported in part by the NSFC (No. 61876195).

\section*{Declaration of competing interest}
The authors declare that they have no known competing financial interests or personal relationships that can have appeared to influence the work reported in this paper.
\section*{Data availability}
The data that support the findings of this study are available from the corresponding author upon reasonable request.


\end{document}